%% file: onebit_pr_ver2.tex
\documentclass[11pt,draftcls,onecolumn]{IEEEtran}

\ifCLASSINFOpdf
\else
   \usepackage[dvips]{graphicx}
\fi
\usepackage{url}
\usepackage{graphicx}
\usepackage{amsmath}
\usepackage{comment}
\usepackage[utf8]{inputenc}
\usepackage[english]{babel}
\usepackage{color}
\usepackage{framed}
\usepackage{cite}
\usepackage{amsmath,amssymb,amsfonts}
\usepackage{graphicx}
\usepackage{textcomp}
\usepackage{cite}
\usepackage{bbold}
\usepackage{ifpdf}
\usepackage{times}
\usepackage{layout}
\usepackage{float}
\usepackage{afterpage}
\usepackage{amsmath}
\usepackage{amstext}
\usepackage{amssymb,bm}
\usepackage{latexsym}
\usepackage{color}
\usepackage{flexisym}
\usepackage{dblfloatfix}    
\usepackage{kantlipsum}
\usepackage{subcaption}
\usepackage{graphicx}
\usepackage{amsmath}
\usepackage{amsthm}
\usepackage{graphicx}
\usepackage{caption}
\usepackage{subcaption}
\usepackage{booktabs}
\usepackage{multicol}
\usepackage{lipsum}

\usepackage{enumitem}
\input{symbols.tex}

\makeatletter
\newcommand*{\rom}[1]{\expandafter\@slowromancap\romannumeral #1@}
\makeatother

\begin{document}

\title{One-Bit Phase Retrieval: \\More Samples Means Less Complexity?}

\author{Arian Eamaz, \IEEEmembership{Student Member, IEEE}, Farhang Yeganegi, and \\ Mojtaba Soltanalian, \IEEEmembership{Senior Member, IEEE}
\thanks{This work was supported in part by the National Science Foundation Grant CCF-1704401.}
\thanks{A. Eamaz, F. Yeganegi and M. Soltanalian are with the Department of Electrical and Computer Engineering, University of Illinois Chicago, Chicago, IL 60607, USA (\emph{Corresponding author: Arian Eamaz}).}
}

\markboth{IEEE TRANSACTIONS ON SIGNAL PROCESSING, 2022 
}
{Shell \MakeLowercase{\textit{et al.}}: Bare Demo of IEEEtran.cls for IEEE Journals}
\maketitle

\begin{abstract}
The classical problem of phase retrieval has found a wide array of applications in optics, imaging and signal processing. In this paper, we consider the phase retrieval problem in a one-bit setting, where the signals are sampled using one-bit analog-to-digital converters (ADCs). A significant advantage of deploying one-bit ADCs in signal processing systems is their superior sampling rates as compared to their high-resolution counterparts. This leads to an enormous amount of one-bit samples gathered at the output of the ADC in a short period of time. We demonstrate that this advantage pays extraordinary dividends when it comes to convex phase retrieval formulations---namely that the often encountered matrix semi-definiteness constraints as well as rank constraints (that are computationally prohibitive to enforce), become redundant for phase retrieval in the face of a growing sample size. Several numerical results are presented to illustrate the effectiveness of the proposed methodologies.
\end{abstract}


\IEEEpeerreviewmaketitle

\section{Introduction}
\IEEEPARstart{P}{hase} retrieval has gained significant interest in applied physics and statistical signal processing communities over the past decades \cite{millane,kim,Fienup93, Krist95, Sarnik,GS1,GS2,fienup1,fienup2,fienup3,candes2013phaselift, candes2014solving,jaganathan2013sparse,candes,chen,TWF,RWF,nayer2021sample,fogel2016phase,jagatap2017fast,bahmani2017phase}. This classical problem manifests as the recovery of an unknown signal solely from phaseless measurements that depend on the signal through a linear observation model. Due to the intrinsic difficulties of the recovery task \cite{sahinoglou1991phase}, recently, there have been many efforts to propose approximate or relaxed versions of the phase retrieval problem in a convex optimization language, particularly via semi-definite programming \cite{candes2013phaselift,candes2015phase}.

Quantization of the signals of interest through analog-to-digital converters (ADCs) is an important task in digital signal processing applications. A very large number of quantization levels is necessary in order to represent the original continuous signal in high-resolution scenarios. The large number of quantization bits, however, can cause a considerable increase in the overall power consumption and the manufacturing cost of ADCs, as well as a reduction in sampling rate \cite{eamaz2021modified}. Such disadvantages have motivated the researchers to investigate the idea of utilizing fewer bits for sampling. \emph{One-bit quantization} is an extreme quantization scenario, in which the signals are compared with given threshold levels at the ADCs, producing sign ($\pm1$) outputs. This enables signal processing equipments to sample at a very high rate, with a considerably lower cost and energy consumption, compared to their counterparts which employ multi-bit ADCs \cite{instrumentsanalog,mezghani2018blind,eamaz2021modified,sedighi2020one}. We further note that one-bit quantization with a fixed threshold (usually zero) can lead to difficulties in the estimation of the signal amplitude. Employing time-varying thresholds, however, has been shown to result in enhanced signal recovery performance in some recent works \cite{eamaz2021modified,eamaz2022covariance,AEamaz2022,qian2017admm,gianelli2016one,wang2017angular,xi2020gridless}.


\subsection{Contributions of the Paper}
While convex formulations of the phase retrieval problem promise a global solution, some of the introduced constraints are computationally costly; including the matrix rank and the positive semi-definite (PSD) constraints. However, we show that if more samples are available, the sheer number of samples can constrain the solution in a less costly manner and make such constraints redundant.
Note that, as mentioned earlier, by employing the one-bit quantization, sampling can be done at significantly higher rates. As a result, the emergence of one-bit sampling techniques paves the way for an investigation on the role of an increased sample size in the phase retrieval problem.

In this paper, we show that the phase retrieval problem can be tackled by taking advantage of the large number of linear observation inequalities that emerge naturally in the one-bit quantization regimen. Instead of considering the often-formulated trace relaxation problem, our approach to one-bit phase retrieval is presented as a randomized Kaczmarz algorithm-based recovery. We present our results on a proper selection of the sufficient number of samples. Furthermore, an algorithm is proposed based on our model to adaptively evaluate the time-varying sampling thresholds. The performance of our approach with an increased sample size is also investigated when noisy measurements are utilized.

\subsection{Organization of the Paper}
Since our approach takes root in convex phase retrieval, Section~\ref{sec_1} is dedicated to a survey of such formulations. In Section~\rom{3}, we will discuss the appearance of one-bit sampling with time-varying thresholds in the phase retrieval context through linear inequality constraints (defining a polyhedron feasible region), as well as the randomized Kaczmarz algorithm (RKA) that can be utilized to recover our desired signal.
To investigate the error recovery of the proposed algorithm, a theorem, which may be useful to select the number of measurement, is presented in Section~\rom{4}. 
Section~\rom{5} is devoted to comparing our method with PhaseLift and its one-bit version in terms of their computational burden. Based on our proposed polyhedron formulation, an algorithm is proposed to obtain the adaptive time-varying thresholds which benefit finding the signal of interest with more accuracy and less computational cost in Section~\rom{6}. In Section~\rom{7}, the noisy measurement scenario is studied owing to its importance in practical applications. Finally, Section~\rom{9} concludes the paper.

\underline{\emph{Notation:}}
We use bold lowercase letters for vectors and bold uppercase letters for matrices. $\mathbb{C}$ and $\mathbb{R}$ represent the set of complex and real numbers, respectively. $(\cdot)^{\top}$ and $(\cdot)^{\mathrm{H}}$ denote the vector/matrix transpose, and the Hermitian transpose, respectively. $I_{N}\in \mathbb{R}^{N\times N}$ is the identity matrix of size $N$. $\operatorname{Tr}(.)$ denotes the trace of the matrix argument. The spectral radius $\rho(\bB)$ of a matrix $\bB$ is defined as a maximum absolute value of its eigenvalues \cite{ortega1990numerical}. The Frobenius norm of a matrix $\bB\in \mathbb{C}^{M\times N}$ is defined as $\|\bB\|_{\mathrm{F}}=\sqrt{\sum^{M}_{r=1}\sum^{N}_{s=1}\left|b_{rs}\right|^{2}}$ where $\{b_{rs}\}$ are elements of $\bB$. The $\ell^{k}$-norm for a vector $\mathbf{b}$ is defined as $\|\bb\|^{k}_{k}=\sum_{i}b^{k}_{i}$. The Hadamard (element-wise) product of two matrices $\bB_{1}$ and $\bB_{2}$ is denoted as $\bB_{1}\odot \bB_{2}$. Additionally, the  Kronecker product is denoted as $\bB_{1}\otimes \bB_{2}$. The vectorized form of a matrix $\bB$ is written as $\operatorname{vec}(\bB)$. $\mathbf{1}_{s}$ is a $s$-dimensional all-one vector. Given a scalar $x$, we define $(x)^{+}$ as $\max\left\{x,0\right\}$. For an event $\mathcal{E}$, $\mathbb{I}_{(\mathcal{E})}$ is the indicator function for that event meaning that $\mathbb{I}_{(\mathcal{E})}$ is $1$ if $\mathcal{E}$ occurs, and $0$ otherwise. $f\asymp g$ means $f$ and $g$ are asymptotically equal. The cumulative distribution function (CDF) of the zero-mean Gaussian process $\bz\sim\mathcal{N}(0,\zeta)$ is given by
\begin{equation}
\label{eq:1bbb}
\Phi(\bz) \triangleq \frac{1}{\sqrt{2\pi}}\int^{z}_{-\infty}e^{-\frac{t^{2}}{2\zeta^{2}}} \,dt.
\end{equation}
To compare two different CDFs, the Hellinger distance may be utilized \cite{davenport20141}, which is defined as
\begin{equation}
\label{eq:1bbbb}
d^{2}_{H}\left(p,q\right) \triangleq \left(\sqrt{p}-\sqrt{q}\right)^{2}+\left(\sqrt{1-p}-\sqrt{1-q}\right)^{2},
\end{equation}
with $p,q\in [0,1]$.

\section{Convex Phase Retrieval: Opportunities and Challenges}
\label{sec_1}
To tackle the phase retrieval problem, many non-convex and local optimization algorithms have been developed over the years. Recently, however, convex programming formulations have come to the fore to yield \emph{global} solutions.
As a case in point, the \emph{PhaseLift} method in \cite{candes2013phaselift} adopts a convex optimization mathematical machinery to tackle the phase retrieval problem, ensuring a near exact recovery of the unknown signal. To do so, PhaseLift relies on a trace-norm relaxation that is used in lieu of the original non-convex rank minimization problem--more on this below. Due to the imposition of the positive semi-definite (PSD) constraint, the PhaseLift problem formulation joins the class of semi-definite programs (SDPs).

Suppose $\mathbf{x}\in \mathbb{C}^{n}$ 
is the discrete signal of interest that is observed linearly through the lens of sensing vectors $\ba_{j}$, with $\left\{\ba_{j}^{\mathrm{H}}\right\}$ constituting the rows of the sensing matrix $\bA \in \mathbb{C}^{m\times n}$. Our goal in phase retrieval is to recover the signal $\mathbf{x}$ from phaseless measurements $y_{j}$ \cite{candes2013phaselift,candes2015phase}:
\begin{equation}
\label{1n}
\begin{aligned}
y_{j}=\left|\ba^{\mathrm{H}}_{j} \mathbf{x}\right|, \quad j \in \mathcal{J}=\left\{1,\cdots,m\right\}.
\end{aligned}
\end{equation}
To ease the mathematical manipulation, one can use the squared version of (\ref{1n}), i.e.,
\begin{equation}
\label{eq:1nn} 
\begin{aligned}
y_{j}^{2} &=\mathbf{x}^{\mathrm{H}} \ba_{j} \ba^{\mathrm{H}}_{j} \mathbf{x}, \\
&=\operatorname{Tr}\left(\ba_{j} \ba^{\mathrm{H}}_{j}\mathbf{x x}^{\mathrm{H}}\right), \\
&=\operatorname{Tr}\left(\bV_{j} \bX\right),
\end{aligned}
\end{equation}
where $\bX=\mathbf{x x}^{\mathrm{H}}$ and $\bV_{j} ={\ba_{j} \ba^{\mathrm{H}}_{j}}$. Based on (\ref{eq:1nn}), the phase retrieval problem can be defined as,
\begin{equation}
\label{eq:1nnn}
\begin{aligned}
\text{find}\quad &\bX\\
\text{s.t.} \quad &\operatorname{Tr}\left(\bV_{j} \bX\right)=y_{j}^{2}, \\
& \operatorname{rank}\left(\bX\right)=1,\\
&\bX \succeq 0.
\end{aligned}
\end{equation}
To have a convex program as \cite{candes2013phaselift}, the problem (\ref{eq:1nnn}) is then relaxed as \cite{candes2013phaselift},
\begin{equation}
\label{eq:1nnnn}
\begin{aligned}
\min_{\bX}\quad &\operatorname{Tr}(\bX)\\
\text{s.t.} \quad &\operatorname{Tr}\left(\bV_{j} \bX\right)=y_{j}^{2}, \\
&\bX \succeq 0.
\end{aligned}
\end{equation}
The linear objective and constraints, along with the PSD constraint, turns (\ref{eq:1nnnn}) to a semi-definite program which is convex \cite{vandenberghe1996semidefinite}. Due to its convexity, there exists a wide array of numerical solvers including the popular Nesterov’s accelerated first-order method to tackle the problem above \cite{candes2015phase,nesterov2003introductory}.

Although, the rank-one and the PSD constraints are deemed necessary to the phase retrieval formulation, they lead to an increased computational cost even in cases where we deal with a convex optimization landscape. To enforce the PSD constraint, a projected gradient method is used in \cite{candes2015phase}, where the approximate solution should be projected onto a PSD cone at each iteration by recovering all eigenvalues and setting the negative eigenvalues to zero, which is quite expensive \cite{candes2015phase}.

An interesting alternative to enforcing the PSD constraint in (\ref{eq:1nnnn}) emerges when one increases the number of samples $m$, and solves the overdetermined linear system of equations with $m\geq n$. By collecting a large number of samples, the linear constraints $\operatorname{Tr}\left(\bV_{j} \bX\right)=y_{j}^{2}$ may actually yield the optimum inside the PSD area $\bX \succeq 0$. As a result of increasing the number of samples, it is possible that the intersection of these hyperplanes will shrink to the optimal point without the need to consider the PSD constraint. However, this idea may face practical limitations in the case of multi-bit quantization systems since ADCs capable of ultra-high rate sampling are difficult and expensive to produce. Moreover, one cannot necessarily expect these constraints to intersect with the PSD cone in such a way to form a finite-volume space before the optimum is obtained \cite{candes2013phaselift}.

As we will show in the next section, by defining the phase retrieval in the one-bit sampling regimen, linear equality constraints are superseded with linear inequalities. Therefore, by increasing the number of samples, we may create a finite-volume space inside the cone $\bX \succeq 0$; making the PSD constraint no longer informative or required. From a practical point of view, one-bit sampling is done efficiently at a very high rate with a significantly lower cost compared to its high-resolution counterpart. Thus, by employing one-bit ADCs, it is practical, indeed natural, to study the game-changing opportunities that emerge in the context of phase retrieval due to the availability of a large number of samples.

\begin{figure*}[t]
	\centering
	\begin{subfigure}[b]{0.32\textwidth}
		\includegraphics[width=1\linewidth]{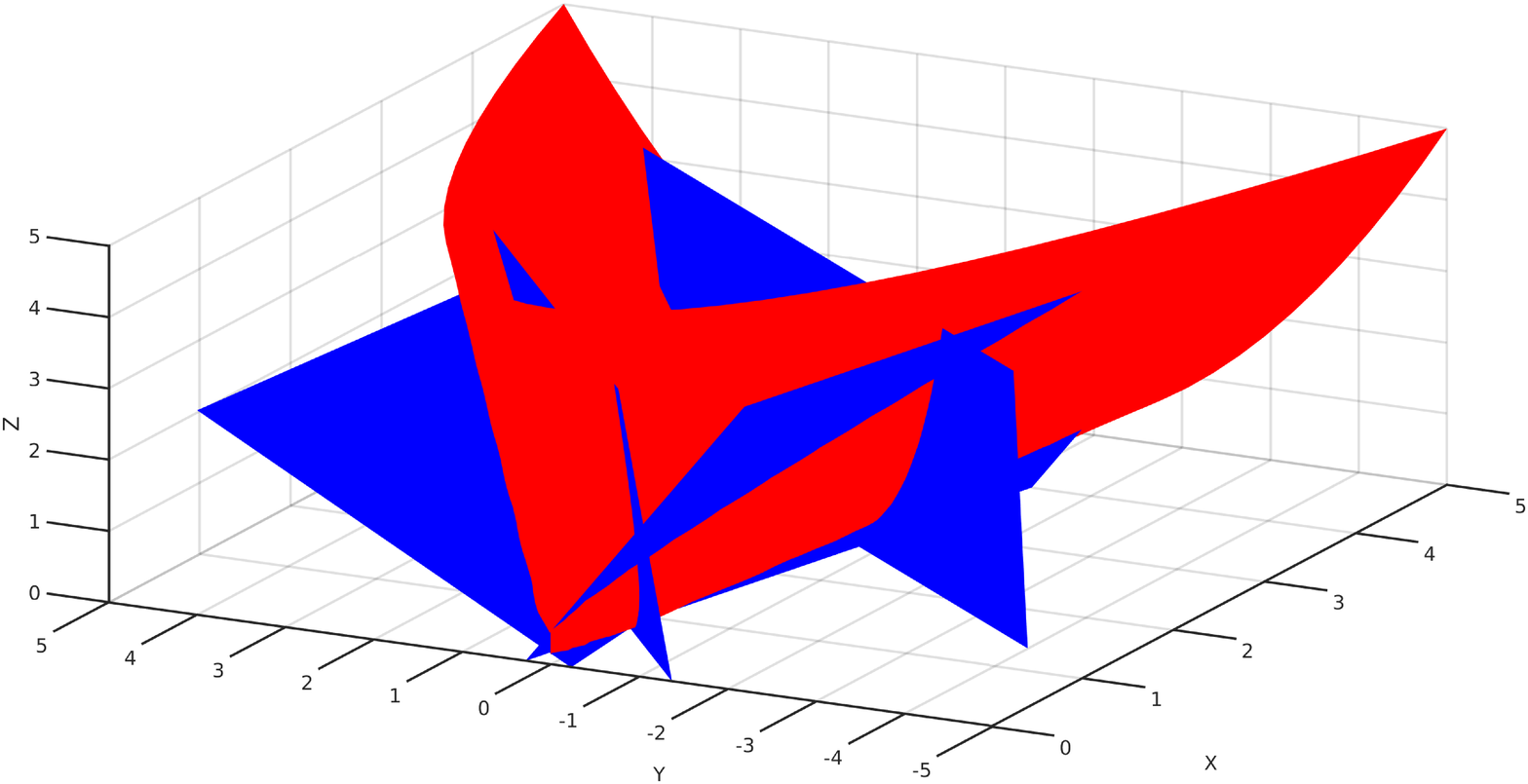}
		\caption{$m=6$}
	\end{subfigure}
	\begin{subfigure}[b]{0.32\textwidth}
		\includegraphics[width=1\linewidth]{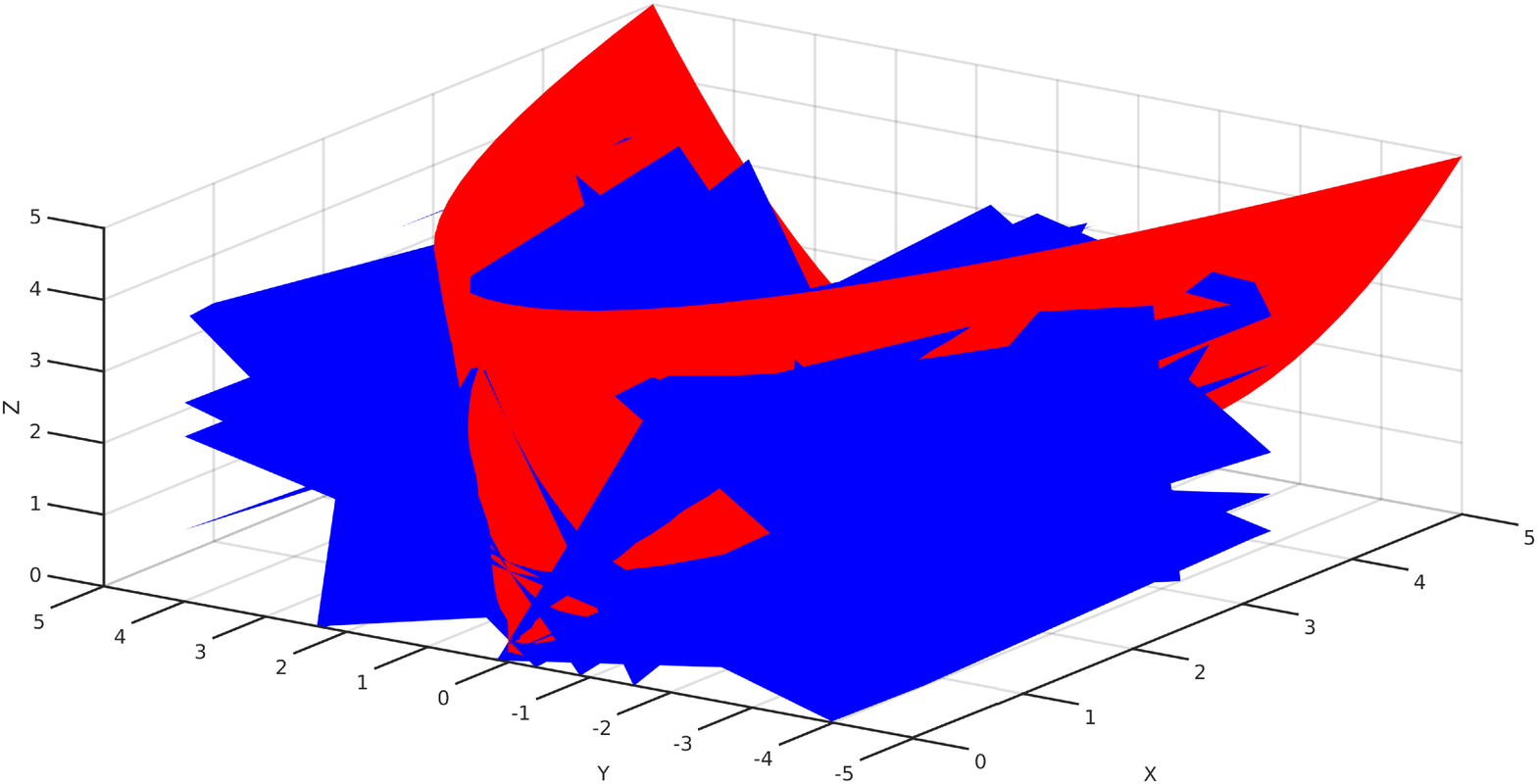}
		\caption{$m=20$}
	\end{subfigure}
	\begin{subfigure}[b]{0.32\textwidth}
		\includegraphics[width=1\linewidth]{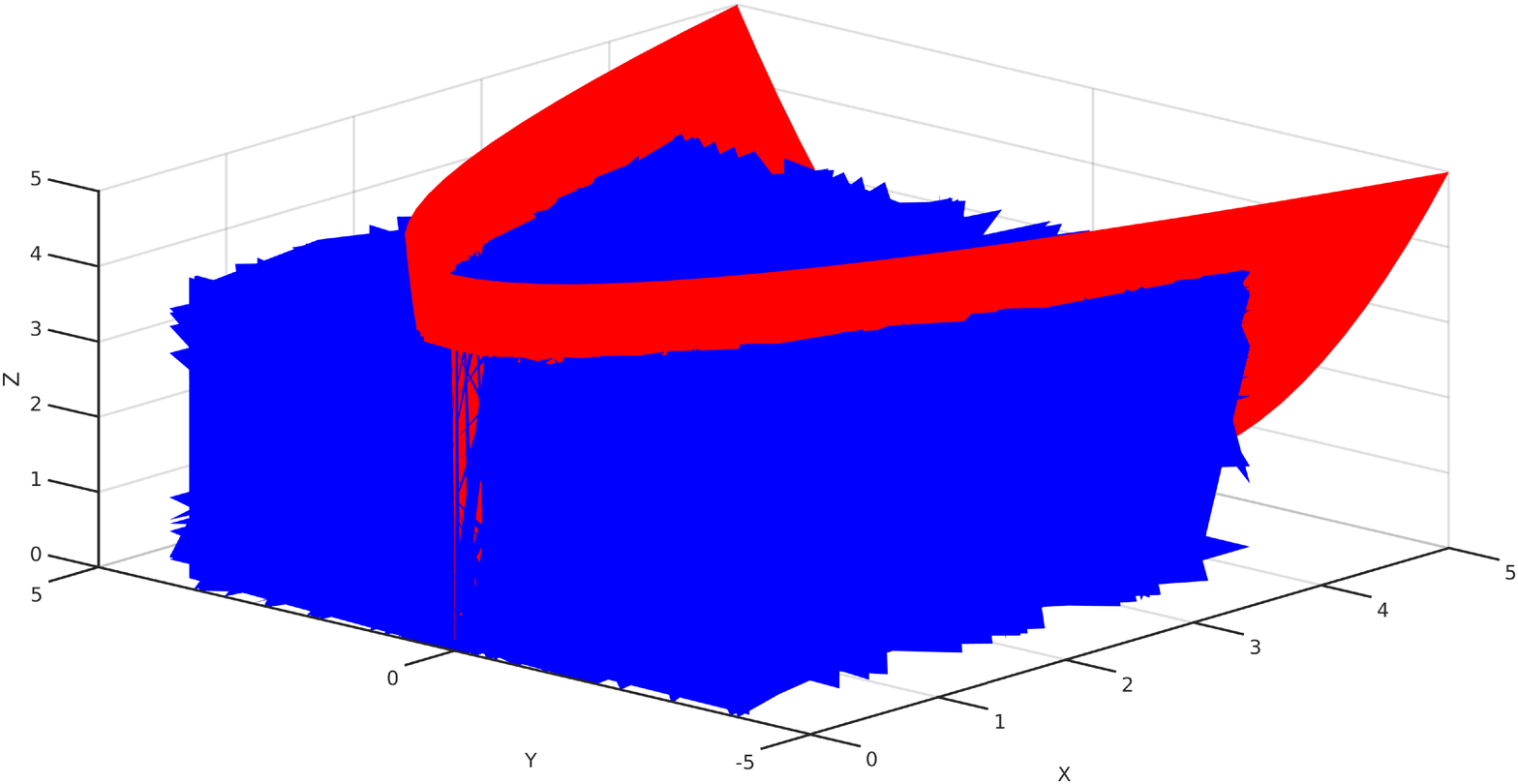}
		\caption{$m=50$}
	\end{subfigure}
		\begin{subfigure}[b]{0.32\textwidth}
		\includegraphics[width=1\linewidth]{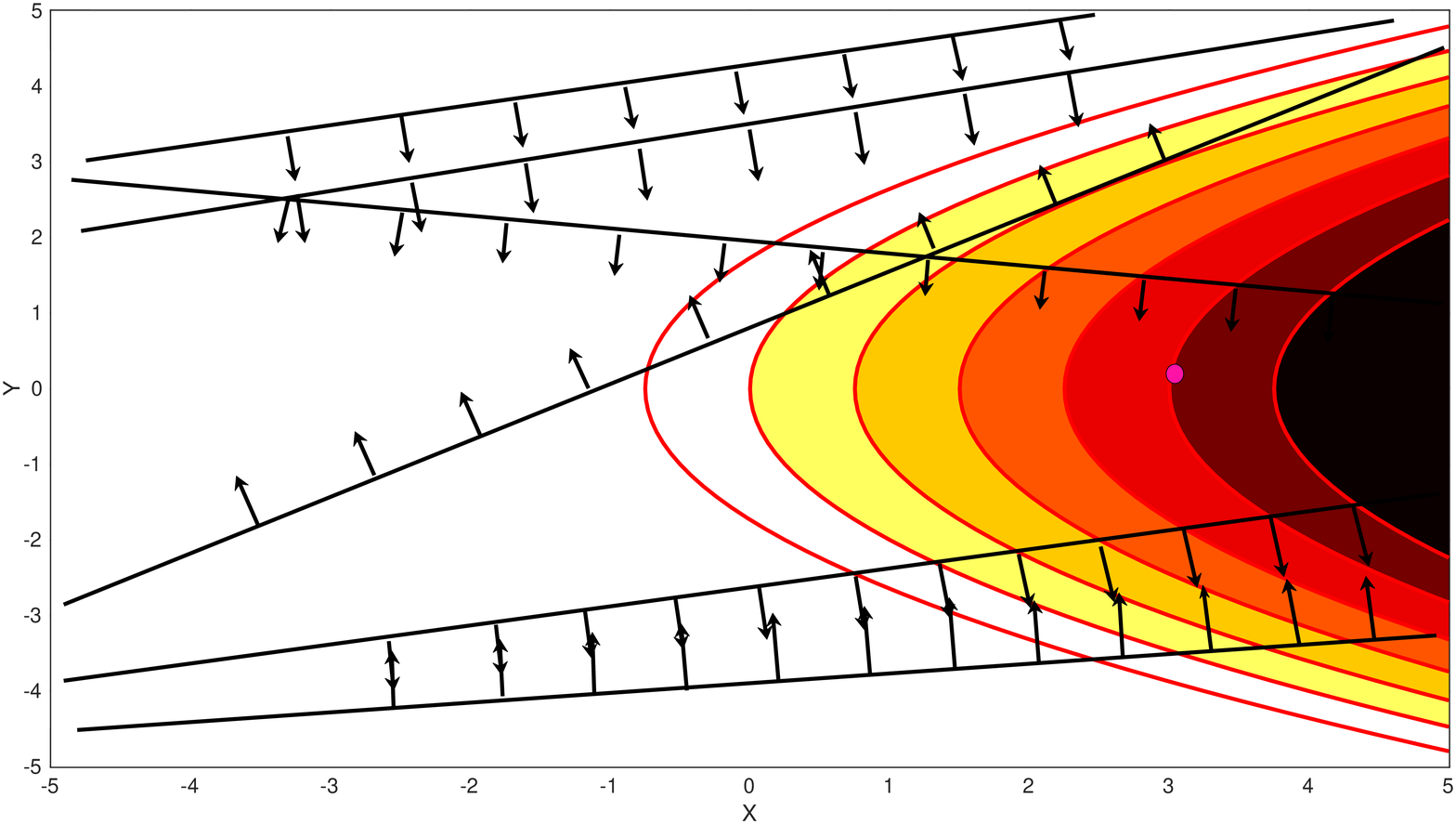}
		\caption{$m=6$}
	\end{subfigure}
		\begin{subfigure}[b]{0.32\textwidth}
		\includegraphics[width=1\linewidth]{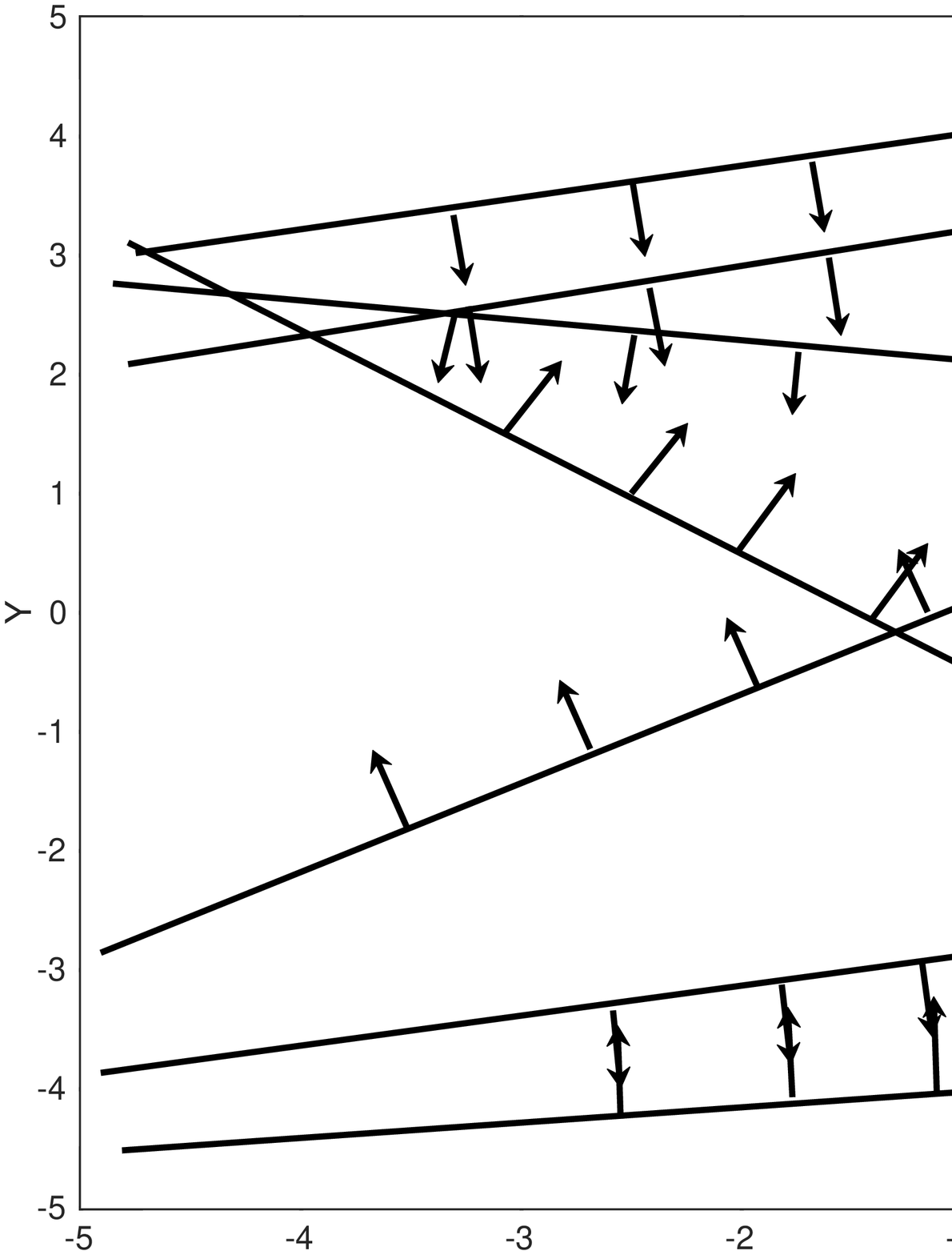}
		\caption{$m=20$}
	\end{subfigure}
		\begin{subfigure}[b]{0.32\textwidth}
		\includegraphics[width=1\linewidth]{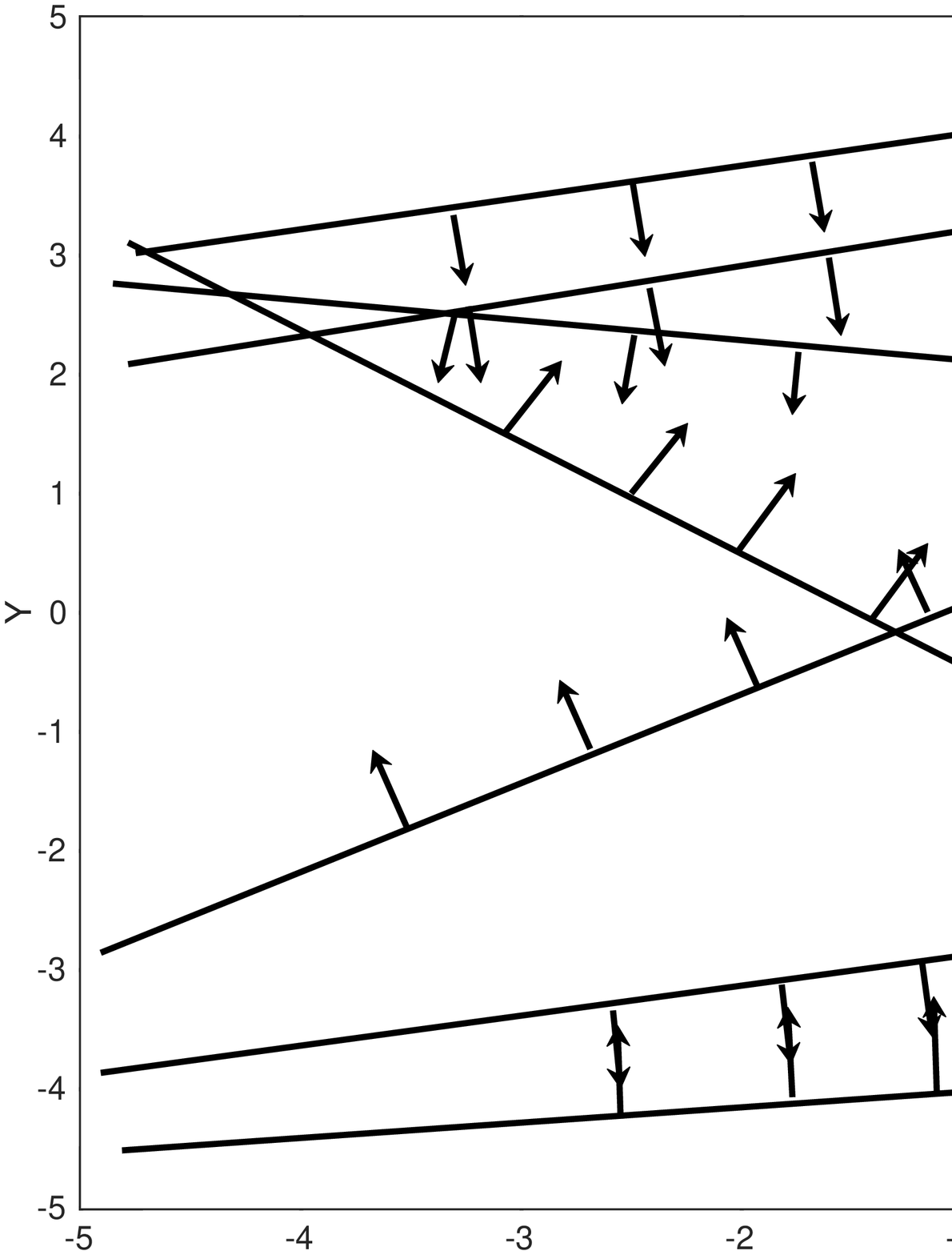}
		\caption{$m=50$}
	\end{subfigure}
	
	\caption{Shrinkage of the polyhedron space (\ref{eq:25}) in blue, ultimately placed within the PSD cone $\bX \succeq 0$ shown with contours and its red boundary, when the number constraints (samples) grows large. The arrows point to the half-space associated with each inequality constraint. The evolution of the feasible regime is depicted with increasing samples in three cases: (a) and (d) small sample-size regime, constraints not forming a finite-value polyhedron; (b) and (e) medium sample-size regime, constraints forming a finite-volume polyhedron, parts of which are outside the cone; (c) and (f) large sample-size regime, constraints forming a finite-volume polyhedron inside the PSD cone, making the PSD constraint redundant. The optimal point representing the signal to be recovered is shown by purple.}
	
	\label{figure_1n}
\end{figure*}

\section{One-Bit Phase Retrieval with Sample Abundance}
\label{sec_3}
As indicated earlier, employing one-bit quantization provides a practical opportunity to address an important question as to whether more samples can mean less complexity in the context of the phase retrieval problem.
We begin our efforts by defining a linear system of inequalities representing the phase retrieval problem in the one-bit quantization system deploying time-varying thresholds leading to the \emph{one-bit phase retrieval} formulation. To recover the desired symmetric positive semi-definite matrix $\bX^{\star}$, we propose an algorithm which relies on exploiting the large number of one-bit sampled data and solves the associated linear system of inequalities by taking advantage of \emph{the randomized Kaczmarz algorithm} (RKA). 

\subsection{Problem Formulation}
\label{one-bit}
In the one-bit sampling scenario, we only observe sign data $\br$, given as
\begin{equation}
\label{eq:3}
\begin{aligned}
\br = \operatorname{sgn}\left(\mathbf{y}-\btau\right),
\end{aligned}
\end{equation}
where $\btau$ is a time-varying\footnote{Note that although we are focusing on temporal sampling, the low cost associated with one-bit ADCs enables the deployment of large arrays of ADCs that are spatially distributed, which is of immediate use in various communications and imaging applications. This paves the way for spatially-varying sampling thresholds, possibly along with time-varying thresholds. Fortunately, the mathematical foundations and algorithms we present in this work can be directly applied to cases where spatially-varying thresholds are used as well.} threshold. Let $e$ denote the phase vector to be recovered. The one-bit phase retrieval problem can be formulated as:
\begin{equation}
\label{eq:1}
\begin{aligned}
\text{find} \quad &\mathbf{y}, \mathbf{x}, \be\\
\mathbf{y}\odot \mathbf{e}&=A\mathbf{x},\\
\mathbf{y}&\in \Psi,
\end{aligned}
\end{equation}
where $\Psi$ is a feasible region created by the one-bit constraints
\begin{equation}
\label{eq:2}
r_{j} \left(y_{j}-\tau_{j}\right)\geq 0, \quad j\in \mathcal{J},
\end{equation}
or equivalently,
\begin{equation}
\Omega\left(\mathbf{y}-\btau\right) \succeq 0,
\end{equation}
with the matrix $\Omega$ defined as $\Omega = \operatorname{diag}\left\{\br\right\}$. Inspired by (\ref{eq:1}), in the following, we present a reformulation of the one-bit phase retrieval problem. Since $y_{j} \geq 0$ based on (\ref{1n}), assuming $\tau_{j} \geq 0$, the following relation holds:
\begin{equation}
\label{eq:10}
\begin{aligned}
y_{j} \lesseqgtr \tau_{j}  \Longleftrightarrow  y^{2}_{j} \lesseqgtr \tau^{2}_{j}.
\end{aligned}
\end{equation}
Therefore, the set of inequalities in (\ref{eq:2}) can be rewritten as
\begin{equation}
\label{eq:11}
\begin{aligned}
r_{j} \left(y_{j}-\tau_{j}\right)\geq 0 &\Longrightarrow  r_{j} \left(y^{2}_{j}-\tau^{2}_{j}\right)\geq 0,\quad j \in \mathcal{J}.
\end{aligned}
\end{equation}
Consequently, one can recast (\ref{eq:1}) in the same spirit as (\ref{eq:1nnn}):
\begin{equation}
\label{eq:5}
\begin{aligned}
\text{find} \quad &\bX\\ \text{s.t.}\quad &r_{j} \left(\operatorname{Tr}\left(\bV_{j}\bX\right)-\tau^{2}_{j}\right) \geq 0,\\
& \operatorname{rank}\left(\bX\right)=1,\\
&\bX \succeq 0.
\end{aligned}
\end{equation}
Moreover, based on (\ref{eq:1nnnn}), the one-bit version of the PhaseLift formulation may be written as
\begin{equation}
\label{eq:6}
\begin{aligned}
\min_{\bX} \quad &\operatorname{Tr}(\bX)\\ \text{s.t.} \quad &r_{j} \left(\operatorname{Tr}\left(\bV_{j}\bX\right)-\tau^{2}_{j}\right) \geq 0,\\ &\bX \succeq 0,
\end{aligned}
\end{equation}
which we refer to as \emph{one-bit PhaseLift} in this paper. It is worth noting that the problem in (\ref{eq:6}) also belongs to the class of semi-definite programs (SDPs). As discussed earlier, in the asymptotic case of the one-bit phase retrieval problem, the PSD constraint may not be required. Moreover, the linear system of inequalities in (\ref{eq:6}) can be reformulated as
\begin{equation}
\label{eq:40000}
\operatorname{Tr}\left(\bV_{j}\bX\right)=\operatorname{vec}\left(\bV^{\top}_{j}\right)^{\top}\operatorname{vec}\left(\bX\right), \quad j\in \mathcal{J},
\end{equation}
where we use the matrix identity \cite{van1996matrix},
\begin{equation}
\operatorname{Tr}\left(\bH^{\top} \bD\right)=\operatorname{vec}(\bH)^{\top} \operatorname{vec}(\bD),
\end{equation}
with $\bH$ and $\bD$ being two arbitrary square matrices. As a result, the constraints imposed in the optimization problem (\ref{eq:6}) can be simplified as
\begin{equation}
\label{eq:9}
\begin{aligned}
\min_{\bX} \quad &\operatorname{Tr}(\bX)\\
\text{s.t.} \quad & \left(\bR\odot \bV\right)\operatorname{vec}\left(\bX\right) \succeq \br \odot \btau^{2},\\
\end{aligned}
\end{equation}
where $\bR=\mathbf{1}_{n^{2}}^{\top}\otimes \br$, and $\bV$ is a matrix with $\operatorname{vec}\left(\bV^{\top}_{j}\right)^{\top}$ as its $j$-th rows ($j\in \mathcal{J}$).

Note that dropping the SDP constraint is not the only advantage of having access to a large number of one-bit sampled data in the context of phase retrieval problem. In fact, we claim that by our approach the rank-one, or its relaxed versions potentially manifested as a trace minimization, also become redundant. To see why, observe that in the asymptotic case of one-bit phase retrieval, the space constrained by the defined inequalities in (\ref{eq:11}), which is a polyhedron, \emph{shrinks} to become contained inside the feasible region in terms of the PSD constraint. However, this shrinking space always contains the globally optimal rank-one solution, with a volume that is decreasing with an increasing number of samples. Thus, instead of the optimization problems in (\ref{eq:6}) and (\ref{eq:9}), we formally define the said polyhedron, i.e.,
\begin{equation}
\label{eq:20}
\begin{aligned}
\mathcal{P} = \left\{\bX \mid r_{j} \left(\operatorname{Tr}\left(\bV_{j}\bX\right)-\tau^{2}_{j}\right) \geq 0,\quad j\in \mathcal{J}\right\},
\end{aligned}
\end{equation}
equivalently restated based on (\ref{eq:40000}) as
\begin{equation}
\label{eq:25}
\mathcal{P}=\left\{\bX \mid r_{j} \operatorname{vec}\left(\bV^{\top}_{j}\right)^{\top}\operatorname{vec}\left(\bX\right) \geq r_{j}\tau^{2}_{j}, \quad j\in \mathcal{J} \right\}.
\end{equation}
A numerical investigation of (\ref{eq:25}) reveals that by increasing the number of samples $m$, the space formed by the intersection of half-spaces (inequality constraints) can fully shrink to the optimal point inside the PSD constraint---see Fig.~\ref{figure_1n} for an illustrative example of this phenomenon. As can be seen in this figure, the black lines representing the linear inequalities form a finite-volume space around the optimal point displayed by the purple circle inside the PSD cone (the elliptical region\footnote{Note that a two-dimensional slice of the three-dimensional PSD cone typically assumes an elliptical form.}) by growing the number of one-bit samples. In (a)/(d), constraints are not enough to create a finite-volume space, whereas in (b)/(e) such constraints can create the desired finite-volume polyhedron space which, however, is not fully inside the PSD cone. Lastly, in (c)/(f), the created finite-volume space shrinks to be fully inside the PSD cone.

To find the signal of interest in the polyhedron (\ref{eq:25}), we use the RKA without enforcing other costly constraints. This is due to the fact that the solution may be efficiently approached by solving the linear system of inequalities presented in (\ref{eq:25}).

Note that two signal models for the phase retrieval problem were introduced in \cite{candes2013phaselift}: (1) The \emph{real-valued model}: the unknown signal $\mathbf{x}$ and $\{\ba_{j}\}$ are real. (2) The \emph{complex-valued model}: the unknown signal $\mathbf{x}$ and $\{\ba_{j}\}$ are complex \cite{candes2013phaselift}. Both settings will be considered in the following proposed algorithm.
\begin{figure*}[t]
	\centering
	\begin{subfigure}[b]{0.45\textwidth}
		\includegraphics[width=1\linewidth]{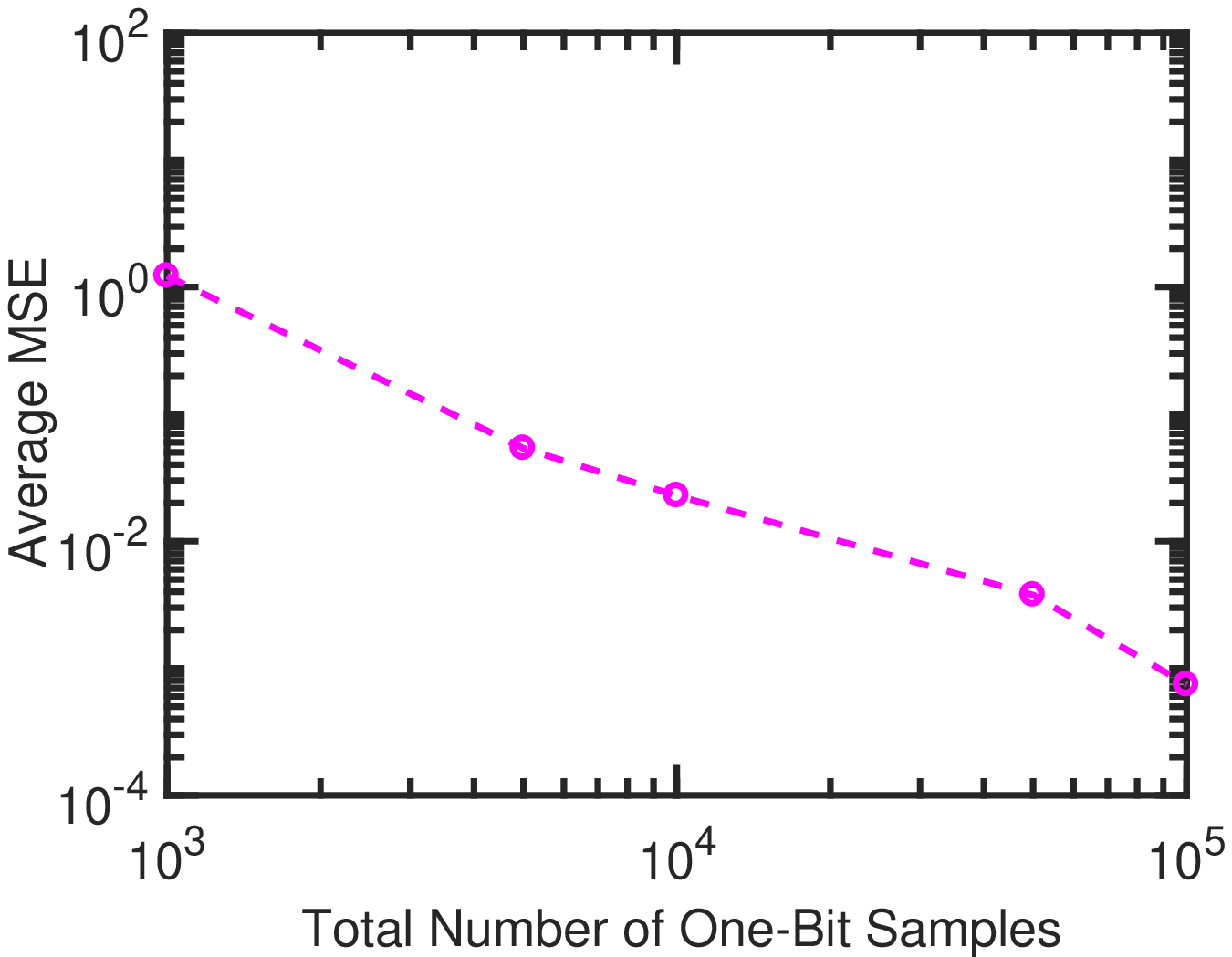}
		\caption{\text{}}
	\end{subfigure}
	\begin{subfigure}[b]{0.45\textwidth}
		\includegraphics[width=1\linewidth]{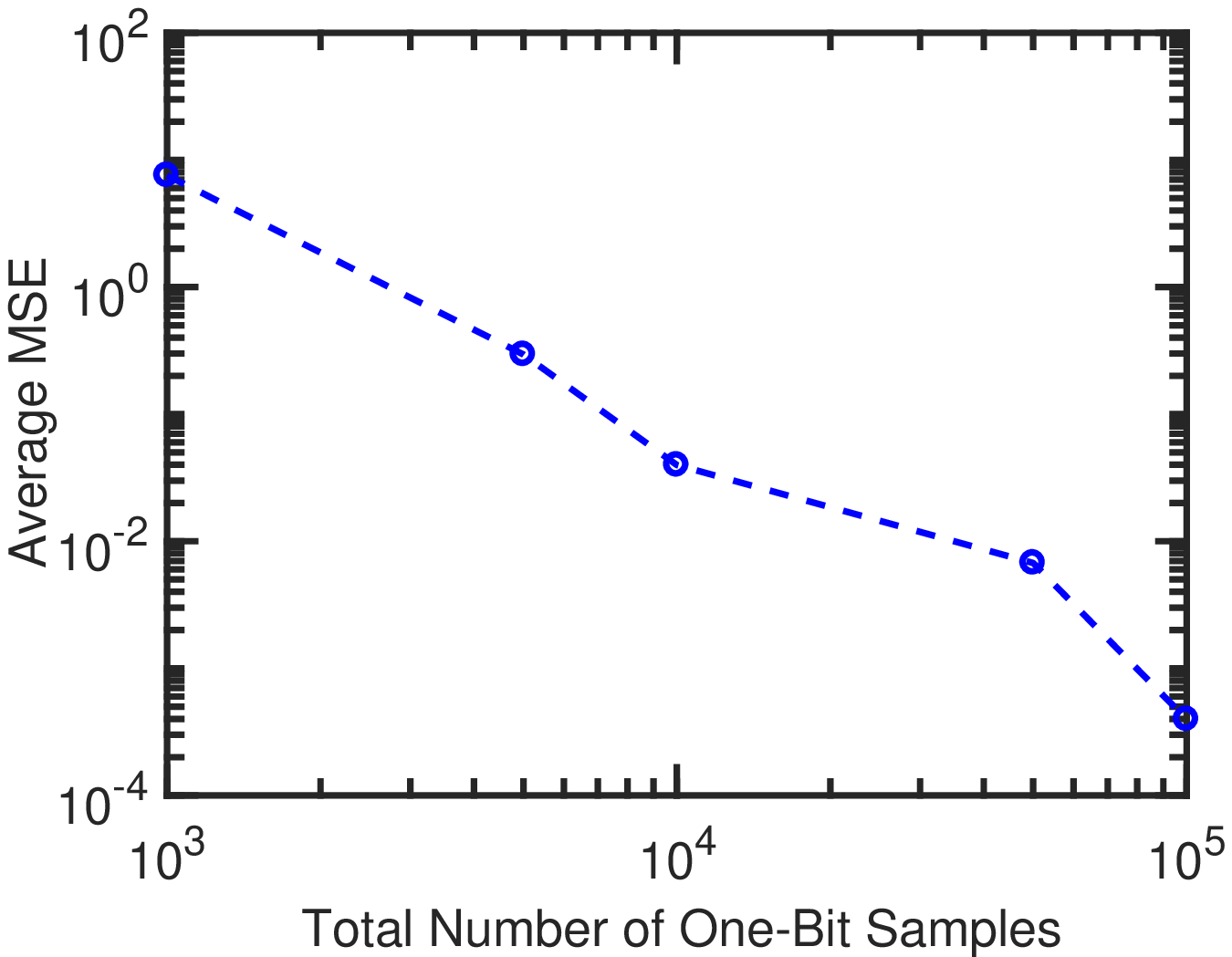}
		\caption{}
	\end{subfigure}
	\caption{Average MSE for signal recovery (in terms of spectral radius of $\bX$) for different one-bit sample sizes with OPeRA: (a) $\mathbf{x}\in \mathbb{R}^{n}$, (b) $\mathbf{x}\in \mathbb{C}^{n}$.}
	\label{figure_2}
\end{figure*}
    \subsection{One-Bit Phase Retrieval Algorithm (OPeRA)}
\label{sec_CRKA}
To recover the desired signal in the one-bit phase retrieval problem, we aim to find a point in the polyhedron (\ref{eq:25}) instead of solving the SDP in (\ref{eq:1nnnn}). As discussed in Section~\ref{one-bit}, when we exploit a large number of samples, the solution of (\ref{eq:25}) is increasingly likely to capture the desired point, i.e. the signal of interest, inside the PSD conical region. The proposed signal recovery relies on the RKA, which is a powerful tool for solving real- or complex-valued linear system of equations, or inequalities through projections \cite{briskman2015block,leventhal2010randomized}.

Accordingly, we propose an algorithm to find the desired matrix $\bX^{\star}$ in (\ref{eq:5}) by (i) using abundant measurements, in order to create the finite-volume space inside the PSD conical region (discussed further in Section~\rom{4}), and (ii) solving (\ref{eq:25}) via the RKA. We name our algorithm the \emph{O}ne-bit \emph{P}has\emph{e} \emph{R}etrieval \emph{A}lgorithm (OPeRA).

The RKA is a \emph{sub-conjugate gradient method} to solve overdetermined linear systems, i.e, $\bC\mathbf{x}\leq\mathbf{b}$ where $\bC$ is a ${m\times n}$ matrix with $m>n$ \cite{leventhal2010randomized,strohmer2009randomized}. Conjugate-gradient methods immediately turn the mentioned inequality to an equality in the following form:
\begin{equation}
\label{bikhod}
  \left(\bC\mathbf{x}-\mathbf{b}\right)^{+}=0,
\end{equation}
and then, approach the solution by the same process as used for systems of equations. Without loss of generality,
consider (\ref{bikhod}) to be a polyhedron:
\begin{equation}
\label{eq:21}
\begin{aligned}
\begin{cases}\bc_{j} \mathbf{x} \leq b_{j} & \left(j \in I_{\leq}\right), \\ \bc_{j} \mathbf{x}=b_{j} & \left(j \in I_{=}\right),\end{cases}
\end{aligned}
\end{equation}
where the disjoint index sets $I_{\leq}$ and $I_{=}$ partition our sample index set $\mathcal{J}$, and $\{\bc_{j}\}$ denote the rows of $\bC$. Based on this problem, the projection coefficient $\beta_{i}$ of the RKA is defined as \cite{leventhal2010randomized,briskman2015block,dai2013randomized}:
\begin{equation}
\label{eq:22}
\beta_{i}= \begin{cases}\left(\bc_{j} \mathbf{x}_{i}-b_{j}\right)^{+} & \left(j \in I_{\leq}\right), \\ \bc_{j} \mathbf{x}_{i}-b_{j} & \left(j \in I_{=}\right).\end{cases}
\end{equation}
Also, the unknown column vector $\mathbf{x}$ is iteratively updated as:
\begin{equation}
\label{eq:23}
\mathbf{x}_{i+1}=\mathbf{x}_{i}-\frac{\beta_{i}}{\left\|\bc_{j}\right\|^{2}_{2}} \bc^{\mathrm{H}}_{j},
\end{equation}
where, at each iteration $i$, the index $j$ is chosen independently at random from the set $\mathcal{J}$, following the distribution
\begin{equation}
\label{eq:24}
P\{j=k\}=\frac{\left\|\bc_{k}\right\|^{2}_{2}}{\|\bC\|_{\mathrm{F}}^{2}}.
\end{equation}

To ensure a limited error, the feasible region in (\ref{eq:25}) cannot be an infinite space in an asymptotic sense. Fortunately, by introducing more samples, the problem can form a polyhedron with a bounded volume containing the desired point. Even more interesting, by adding more inequality constraints in (\ref{eq:25}), the shrinkage of the said polyhedron will put a downward pressure on the error between the desired and recovered points (each informative sample will shrink this space). We will show that by increasing the number of constraints and effective sampling, this error approaches zero. Moreover, as a result of using an overdetermined linear system of inequalities, the convergence of the RKA is guaranteed \cite{leventhal2010randomized,briskman2015block}.

It is worth noting that, in our problem, we only have the inequality partition $I_{\leq}$. Herein, the row vectors $\{\bc_{j}\}$ and the scalars $\{b_{j}\}$ used in the RKA (\ref{eq:21})-(\ref{eq:24}) are $-\left\{r_{j} \operatorname{vec}\left(\bV^{\top}_{j}\right)^{\top}\right\}$ and $-\left\{r_{j}\tau^{2}_{j}\right\}$, respectively.
\begin{figure*}[t]
	\centering
	\begin{subfigure}[b]{0.45\textwidth}
		\includegraphics[width=1\linewidth]{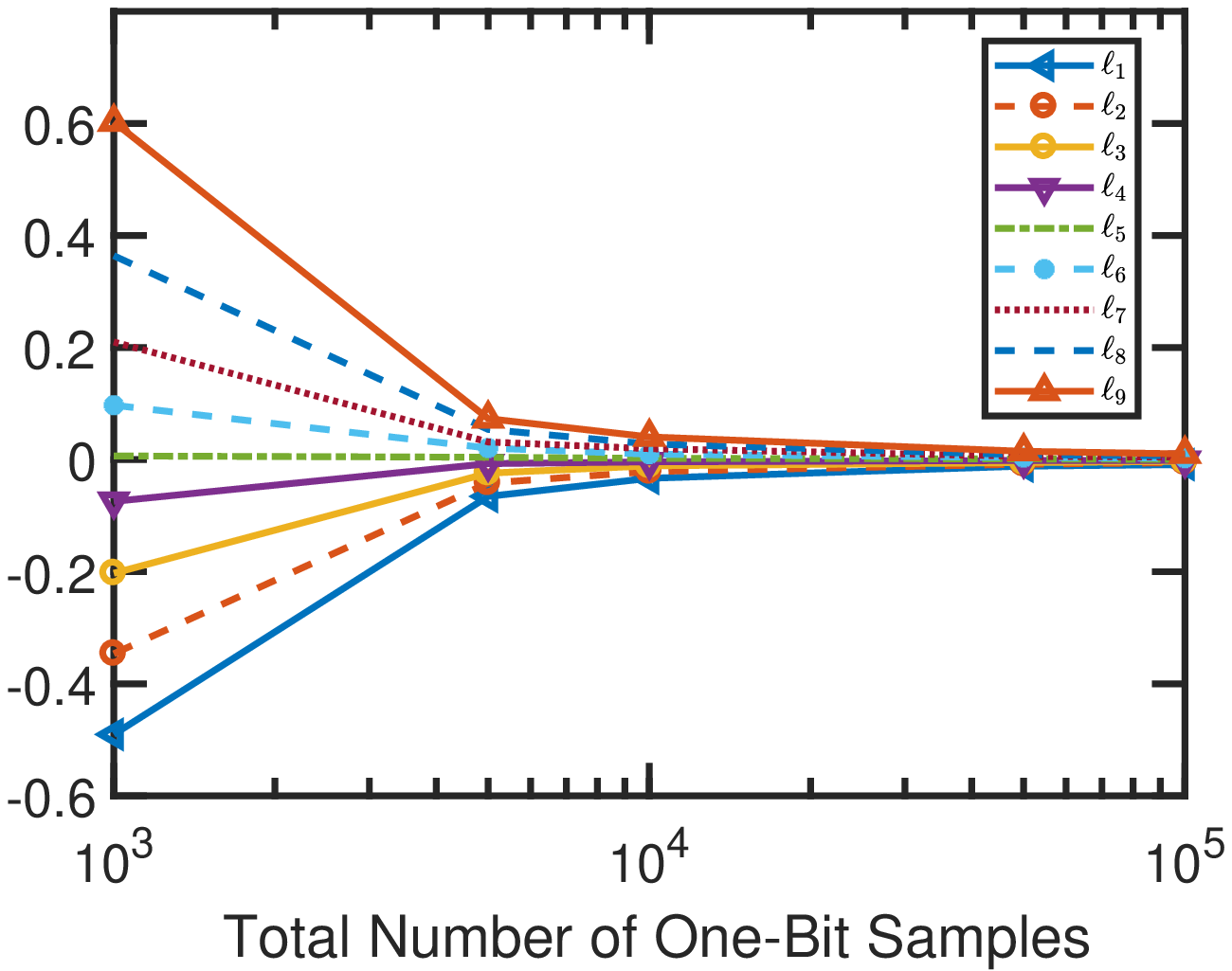}
		\caption{\text{}}
	\end{subfigure}
	\begin{subfigure}[b]{0.45\textwidth}
		\includegraphics[width=1\linewidth]{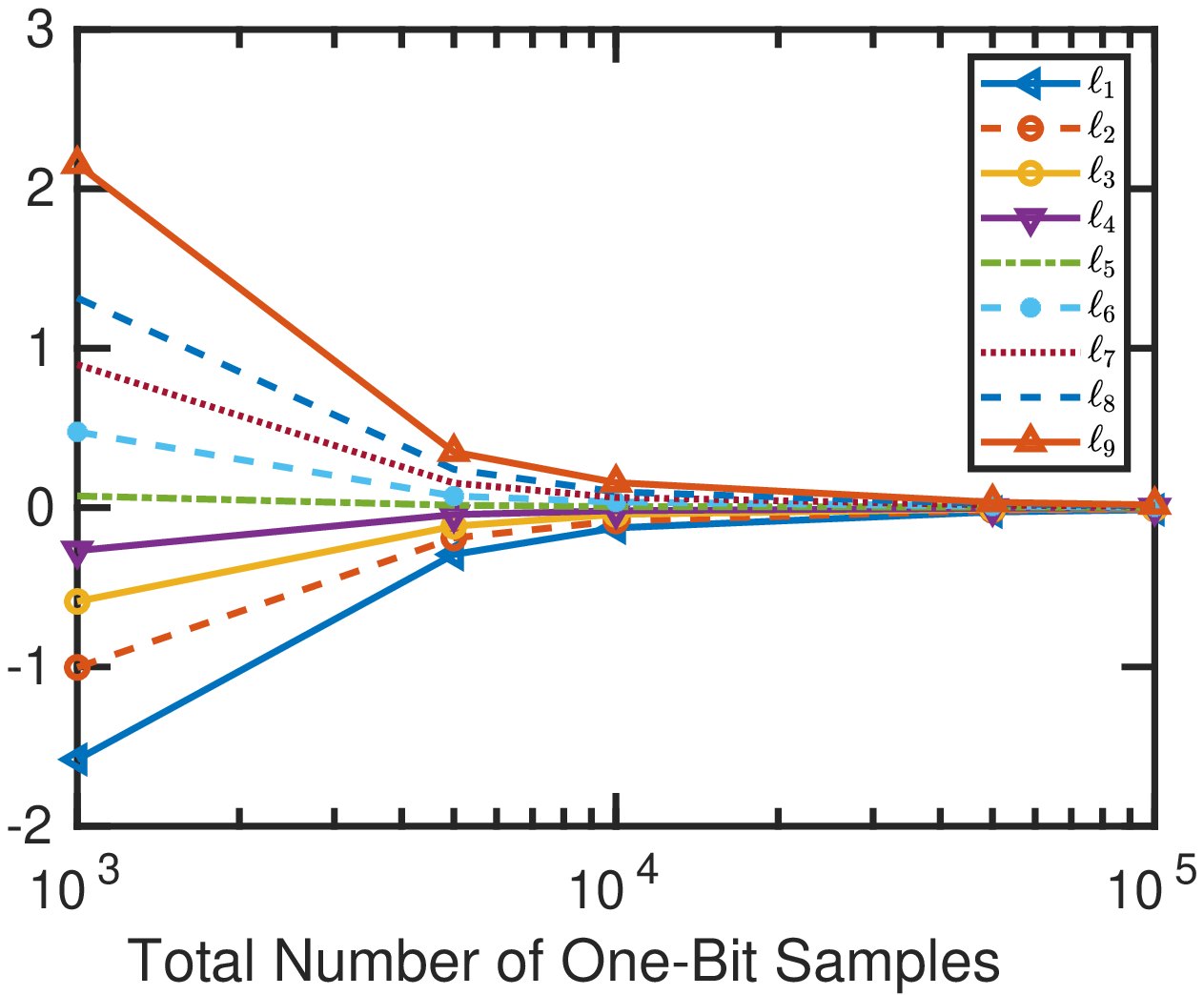}
		\caption{}
	\end{subfigure}
	\caption{The eigenvalues $\{\ell_{i}\}$ of $\bar{\bX}$ (excluding the maximum eigenvalue) averaged over $10$ experiments: (a) $\mathbf{x}\in \mathbb{R}^{n}$, (b) $\mathbf{x}\in \mathbb{C}^{n}$. Deploying a large number of samples leads to obtaining an $\bX$ that is ``increasingly" rank-one PSD.}
	\label{figure_3}
\end{figure*}

\begin{figure*}[t]
	\hspace{.3cm}
	\begin{subfigure}[b]{0.45\textwidth}
		\includegraphics[width=1\linewidth]{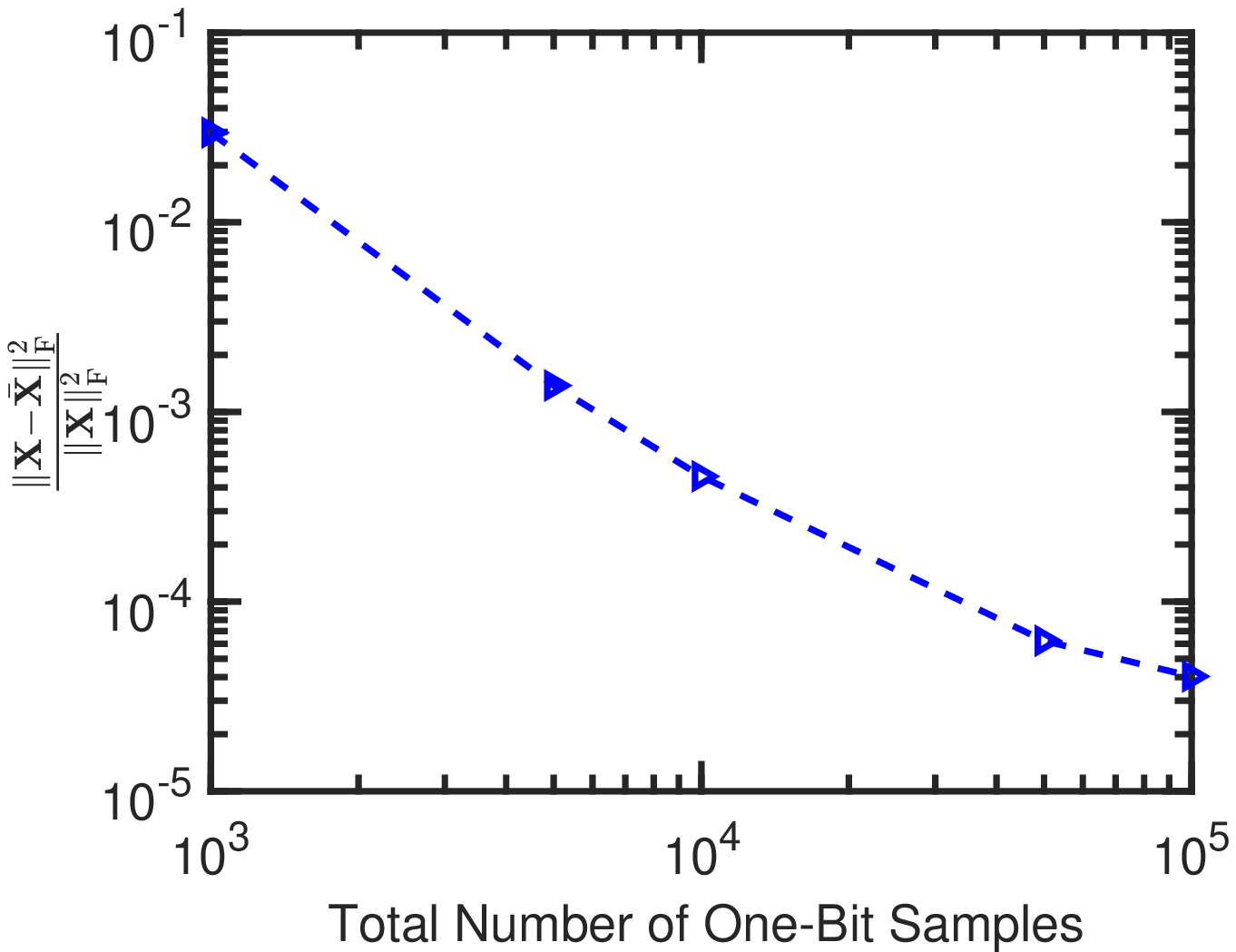}
		\caption{\text{}}
	\end{subfigure}
	\begin{subfigure}[b]{0.455\textwidth}
		\includegraphics[width=1\linewidth]{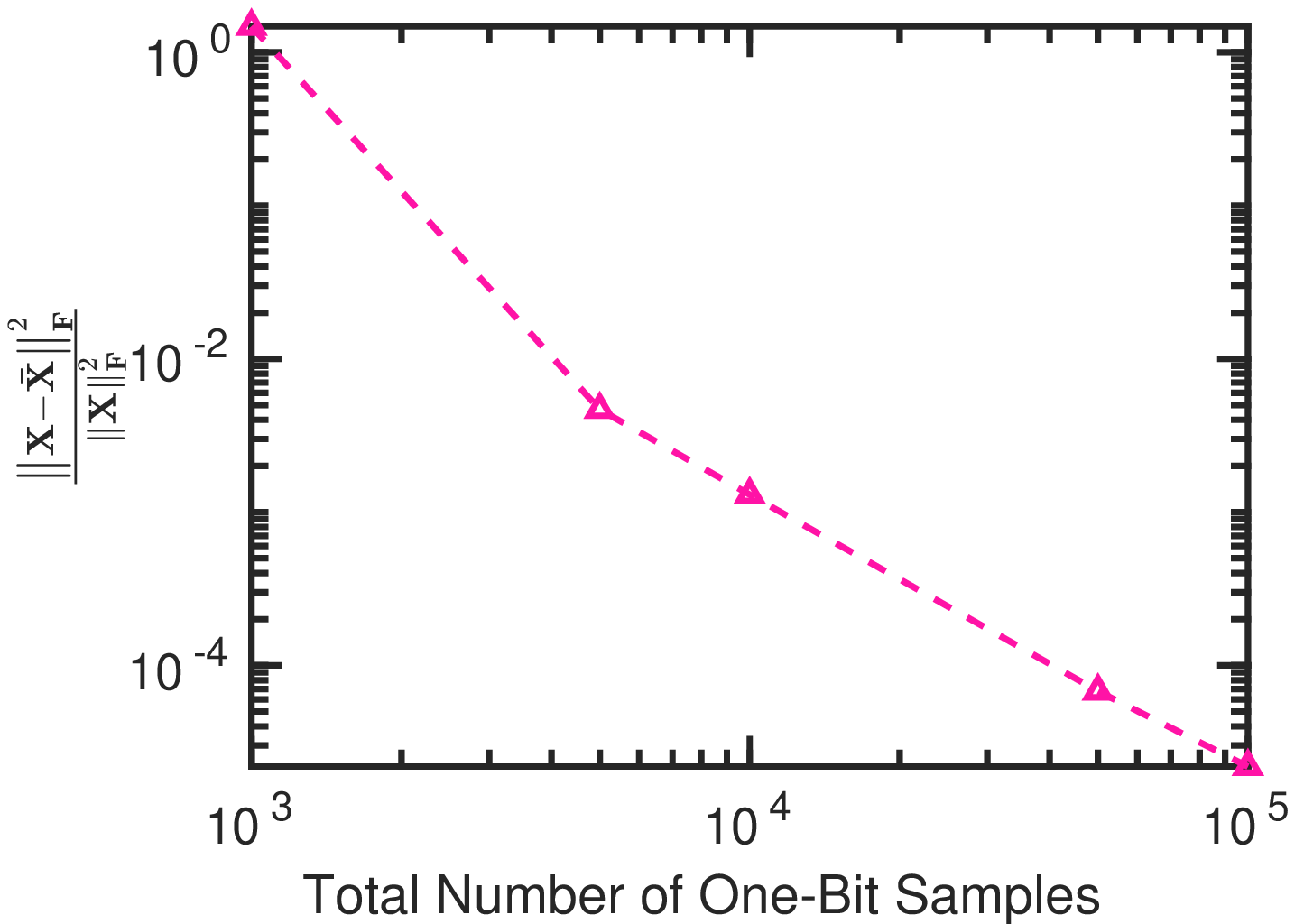}
		\caption{}
	\end{subfigure}
	\caption{Average NMSE for the Frobenius norm error between the desired matrix $\bX^{\star}$ and its recovered version $\bar{\bX}$ for different one-bit sample sizes with the RKA applied to (\ref{eq:25}) when (a) $\mathbf{x}\in \mathbb{R}^{n}$, (b) $\mathbf{x}\in \mathbb{C}^{n}$.}
	\label{figure_3_frob}
\end{figure*}



\subsection{Numerical Illustrations for OPeRA}
\label{NUM}
We numerically examine the effect of growing the sample size in the recovery of the desired matrix $\bX^{\star}$ in OPeRA. We will use the spectral radius metric which is particularly informative in the recovery of rank-one matrices. Note that for a positive semi-definite matrix such as our desired matrix $\bX^{\star}$, the spectral radius is equal to the Frobenius norm of the matrix \cite{van1996matrix}. In all experiments, the input signals were generated as $\mathbf{x}\sim \mathcal{N}\left(\mathbf{0},I_{n}\right)+j\mathcal{N}\left(\mathbf{0},I_{n}\right)$ for the complex-valued model, and $\mathbf{x}\sim \mathcal{N}\left(\mathbf{0},I_{n}\right)$, for the real-valued model. For both models, the rows of the sensing matrix $\bA$ were generated as $\ba_{j} \sim \mathcal{N}\left(\mathbf{0},I_{n}\right)$. Accordingly, we made use of the time-varying thresholds $\btau \sim \operatorname{Lognormal}\left(0,1\right)$. We define the experimental mean square error (MSE) between the true spectral radius $\rho\left(\bX\right)$ and its estimate $\rho\left(\bar{\bX}\right)$ as
\begin{equation}
\label{eq:27}
\mathrm{MSE}\triangleq\frac{1}{E}\sum^{E}_{e=1}\left|\rho_{e}\left(\bX^{\star}\right)-\rho_{e}\left(\bar{\bX}\right)\right|^{2},
\end{equation}
where $E$ is the number of experiments. Each presented data point is averaged over $10$ experiments. The results are obtained for the number of samples $m\in \left\{1000, 5000, 10000, 50000, 100000\right\}$.

Fig.~\ref{figure_2} appears to confirm the possibility of recovering the spectral radius of $\bX^{\star}$ from the large number of one-bit sampled data with time-varying thresholds by OPeRA for the real-valued and the complex-valued models, respectively. As expected, the performance of the recovery will be significantly enhanced as the number of one-bit samples grows large. In all experiments, the obtained maximum eigenvalue is positive, which is equal to the spectral radius of the recovered matrix.

An important part of our work is to show that our method can recover an $\bX$ that is ``increasingly" rank-one and PSD by growing the number of one-bit sampled data. To do so, at first, we show in Fig.~\ref{figure_2} that the maximum eigenvalue is accurately recovered. Next, we present that all eigenvalues $\{\ell_{i}\}$ of the recovered matrix $\bar{\bX}$ except the maximum eigenvalue approach zero by increasing the number of one-bit sampled data. Fig.~\ref{figure_3} appears to confirm this claim for both real-valued and complex-valued models. The presented results are averaged over $10$ experiments and the eigenvalues are arranged in descending order.

To further investigate the effectiveness of OPeRA in both real-valued and complex-valued models, Fig.~\ref{figure_3_frob} illustrates the squared Frobenius norm of the error normalized by the squared Frobenius norm of the desired matrix $\bX^{\star}$, defined as
\begin{equation}
\label{eq:4000}
\mathrm{NMSE}\triangleq\frac{\left\|\bX^{\star}-\bar{\bX}\right\|^{2}_{\mathrm{F}}}{\left\|\bX^{\star}\right\|^{2}_{\mathrm{F}}},
\end{equation}
where the presented results are averaged over $10$ experiments. Fig.~\ref{figure_3_frob} appears to confirm that the performance of the recovery is enhanced by increasing the number of one-bit samples.

\section{Bounding the Recovery Error}
\label{error_bound}
In this section, we derive the convergence rate of our proposed algorithm in its search for the optimal point in the PSD cone. Moreover, an upper bound for the recovery error $\mathbb{E}\left\{\left\|\bX_{i}-\bX^{\star}\right\|_{\mathrm{F}}^{2}\right\}$ will be introduced. This bound will be leveraged to find a lower bound on the number of measurements $m$; the critical role of which was readily discussed in Section~\ref{sec_3}.

\subsection{Chernoff Bound Analysis for OPeRA}
\label{Chernoff_bound}
We first investigate the convergence of OPeRA through a probabilistic lens. Define the distance between the optimal point $\bX^{\star}$ and the $j$-th hyperplane presented in (\ref{eq:25}) as
\begin{equation}
\label{distance}
\begin{aligned}
H_{j}\left(\bX^{\star},\bX_{i}\right) &= \left|\operatorname{Tr}\left(\bV_{j}\bX_{i}\right)-\operatorname{Tr}\left(\bV_{j}\bX^{\star}\right)\right|,\\&=\left|\operatorname{Tr}\left(\bV_{j}\left(\bX_{i}-\bX^{\star}\right)\right)\right|,\quad j\in\mathcal{J},
\end{aligned}
\end{equation}
where $\bX_{i}$ is the solution from the RKA iterations. From an intuitive point of view, it is easy to observe that by generally reducing the distances between $\bX^{\star}$ and the constraint-associated hyperplanes, the possibility of \emph{capturing} the optimal point is increased.

Suppose $\mathcal{K}$ is the cardinality of the set of the hyperplanes presented in (\ref{eq:25}), a portion of the whole sample size, i.e., $\mathcal{K}\leq m$ , which will effectively form a polyhedron inside the PSD cone around the optimal point $\bX^{\star}$ if the number of samples is sufficient. The average of the distances $\left\{H_{j}\right\}$ around the optimal point is obtained as
\begin{equation}
\label{average}
\mathcal{E}\left(\bX^{\star}\right) = \frac{1}{\mathcal{K}}\sum^{\mathcal{K}}_{j=1}H_{j}\left(\bX^{\star},\bX_{i}\right),\quad \mathcal{K}\subseteq\mathcal{J}.
\end{equation}
For a specific sample size $m$, when the area of the finite-volume space around the optimal point is reduced, $\mathcal{E}\left(\bX^{\star}\right)$ is diminished as well. Define the overall average distance $T_{\text{ave}}$ as
\begin{equation}
\label{ave}
\begin{aligned}
T_{\text{ave}} &= \frac{1}{m}\sum^{m}_{j=1}H_{j}\left(\bX^{\star},\bX_{i}\right),\quad m\in\mathcal{J}.
\end{aligned}
\end{equation}
Increasing the number of samples leads to smaller values of $\mathcal{E}\left(\bX^{\star}\right)$ and $T_{\text{ave}}$. Therefore, the possibility of creating the finite-volume space around the desired point increases in the asymptotic sample-size scenario, with $m\geq m^{\star}$, where $m^{\star}$ is the minimal measurement size. The \emph{Chernoff Bound} \cite{chernoff1952measure,chernoff2014career} can shed light on this phenomenon as illustrated bellow.
\begin{theorem}
\label{theorem_0}
Consider the distances $\left\{H_{j}\left(\bX^{\star},\bX_{i}\right)\right\}$ between the desired point $\bX^{\star}$ and the hyperplanes of the polyhedron defined in (\ref{eq:25}) to be i.i.d. random variables. 
\begin{itemize}
    \item The Chernoff bound of $T_{\text{ave}}$ in (\ref{ave}) is given by
\begin{equation}
\label{eq:theorem_cher}
\operatorname{Pr}\left(T_{\text{ave}}=\frac{1}{m}\sum^{m}_{j=1}H_{j}\left(\bX^{\star},\bX_{i}\right)\leq a\right)\geq 1-\inf_{t\geq 0}\frac{M_{T}}{e^{ta}},
\end{equation}
where $M_{T}$ is the moment generating function (MGF) of $T_{\text{ave}}$, given as
\begin{equation}
\label{eq:psi}
M_{T} = \left(1+t\frac{\mu^{(1)}_{H_{j}}}{m}+\cdots+t^{\kappa}\frac{\mu^{(\kappa)}_{H_{j}}}{\kappa!m^{\kappa}}+\mathcal{R}\left(m\right)\right)^{m},
\end{equation}
with $\mu^{(\kappa)}_{H_{j}}=\mathbb{E}\left\{H^{\kappa}_{j}\right\}$, and $\mathcal{R}$ denoting a bounded reminder associated with truncating the Taylor series expansion of $M_{T}$.
\item $M_{T}$ is decreasing with an increasing sample size in the sample abundance scenario, leading to an increasing lower bound in (\ref{eq:theorem_cher}).
\end{itemize}
\end{theorem}
\begin{proof}
The MGF of $T_{\text{ave}}$ is given by
\begin{equation}
\label{eq:mgf_T}
\begin{aligned}
M_{T} = \mathbb{E}\left\{e^{t\frac{1}{m}\sum^{m}_{j=1}H_{j}}\right\}&= \prod^{m}_{j=1}\mathbb{E}\left\{e^{t\frac{H_{j}}{m}}\right\}\\&= \mathbb{E}\left\{e^{t\frac{H_{j}}{m}}\right\}^{m}.
\end{aligned}
\end{equation}
By using the Taylor series expansion, one can write
\begin{equation}
M_{T}=\mathbb{E}\left\{1+t\frac{H_{j}}{m}+t^{2}\frac{H^{2}_{j}}{2m^{2}}+t^{3}\frac{H^{3}_{j}}{6m^{3}}+\cdots\right\}^{m},
\end{equation}
which leads to the formulation (\ref{eq:psi}). Note that due to the fact that the distances $\left\{H_{j}\right\}$ can be considered to be finite values, their moments always exist. 

It is straightforward to verify that $M_{T}$ is an analytic function and that $\left|\mathcal{R}\left(m^{\star}\right)\right|$ is bounded. Let $t_{0}$ denote the value of $t$ making the upper bound in (\ref{eq:theorem_cher}) infimum. To prove that $M_{T}$ is a decreasing function in the asymptotic sample-size case ($m>m^{\star}$), we use the \text{Padé} approximation (PA) which can asymptotically approximate $M_{T}$ with a rational function of given order through the \emph{moment matching} technique as follows \cite{eamaz2021modified,eamaz2022covariance,AEamaz2022}:
\begin{equation}
\label{pade_1}
\forall~m>m^{\star}: ~M_{T}= \left(1+\cdots+t^{\kappa}\frac{\mu^{(\kappa)}_{H_{j}}}{\kappa!m^{\kappa}}\right)^{m}\asymp\frac{a_{0}+\frac{a_{1}}{m}}{b_{0}+\frac{b_{1}}{m}}, 
\end{equation}
where $\left\{a_{0}, a_{1}, b_{0}, b_{1}\right\}$ are the PA coefficients as given in \cite{eamaz2021modified}. The above rational approximation\footnote{Considering the Taylor series expansion of $M_{T}$, the PA with the utilized orders presented in (\ref{pade_1}) will approximate $M_{T}$ with an error in the order of $\mathcal{O}(m^{-3})$ for $m>m^{\star}$.} is a decreasing function; a fact that can be verified by taking its first derivative with respect to $m$. The negativity of the derivative is easily concluded by observing that
\begin{equation}
\label{pade_2}
\sigma^{2}_{H_{j}} = \mu^{(2)}_{H_{j}}-\left(\mu^{(1)}_{H_{j}}\right)^{2} \geq 0,
\end{equation}
where $\sigma^{2}_{H_{j}}$ is the non-negative variance of the random variable $H_{j}$.

\end{proof}

\subsection{Recovery Error Upper Bound for OPeRA}
\label{OPeRA_bound}
In order to find the signal of interest in the polyhedron (\ref{eq:25}), we are utilizing the RKA which leads to the following convergence bound \cite{polyak1964gradient,strohmer2009randomized,briskman2015block,leventhal2010randomized}:
\begin{equation}
\label{bound1}
\mathbb{E}\left\{\left\|\mathbf{x}_{i}-\mathbf{x}^{\star}\right\|_{2}^{2}\right\} \leq q^{i}\left\|\mathbf{x}_{0}-\mathbf{x}^{\star}\right\|_{2}^{2},
\end{equation}
where $\mathbf{x}^{\star}$ is a desired point, $q \in (0,1)$ is a function of the \emph{condition number} of the matrix $\bC$, and $i$ is the number of required iterations for the RKA. In our problem, $\mathbf{x}_{i}=\operatorname{vec}\left(\bX_{i}\right)$, $\mathbf{x}^{\star}=\operatorname{vec}\left(\bX^{\star}\right)$,  and $\bC=\bR\odot\bV$. Therefore, the right-hand side of (\ref{bound1}) may be recast as
\begin{equation}
\label{bound0}
\mathbb{E}\left\{\left\|\operatorname{vec}\left(\bX_{i}\right)-\operatorname{vec}\left(\bX^{\star}\right)\right\|_{2}^{2}\right\}=\mathbb{E}\left\{\left\|\bX_{i}-\bX^{\star}\right\|_{\mathrm{F}}^{2}\right\}.
\end{equation}
It is clear from (\ref{bound1}) that by using a well-chosen initial point $\mathbf{x}_{0}$ or by increasing the number of iterations $i$, the recovery error can be further contained. Nevertheless, in the proposed recovery approach, it is deemed necessary to have the sufficient number of samples (inequalities) in order to guarantee a finite-volume feasible region and a bounded recovery error. Once our search area is located inside the PSD cone, we may effectively employ (\ref{bound1}) for the convergence rate. The convergence rate of the RKA is useful when we have a linear system of inequalities. On the other hand, in the one-bit phase retrieval, the main constraints, i.e. the rank-one and the PSD, are non-linear and they may be considered to be redundant by deploying the enough number of samples. Thus, (\ref{bound1}) is insufficient to present the convergence rate of OPeRA. We can make (\ref{bound1}) relevant to the one-bit phase retrieval problem by taking a penalty function into consideration:
\begin{equation}
\label{bound2}
\begin{aligned}
\mathbb{E}\left\{\left\|\bX_{i}-\bX^{\star}\right\|_{\mathrm{F}}^{2}\right\}& \leq q^{i}\left\|\bX_{0}-\bX^{\star}\right\|_{\mathrm{F}}^{2}+\Psi\left(m\right),
\end{aligned}
\end{equation}
where $\Psi(.)$ is a decreasing function in the abundance sample-size regime, such that if the number of samples is enough to satisfy the PSD constraint, the penalty function approaches zero $\Psi\left(m\right)\rightarrow 0$. Based on our discussion in Section~\ref{Chernoff_bound}, a good example for $\Psi\left(m\right)$ can be $M_{T}-M_{\infty}$, where $M_{\infty}=\lim_{m\rightarrow \infty} M_{T}$.

To find a bound for the sufficient number of measurements $m$ to create a finite-volume space inside the PSD cone and the penalty function starts to be zero, we utilize the tail function of the penalty given by
\begin{equation}
\label{tail}
\Psi^{\star}( m)=\Psi(m)\mathbb{I}_{\left(m\geq m^{\star}\right)}.
\end{equation}
Tails of decreasing functions may be asymptotically approximated by an exponential function \cite{chernoff2014career}. Mathematically, this may be expressed as
\begin{equation}
\label{eq:6585}
\exists~\epsilon_{0}, \gamma_{1} \in \mathbb{R}^{+},\quad\sup_{m} \left|\Psi^{\star}(m) - \epsilon_{0}e^{-\gamma_{1} m}\right|<\varepsilon,
\end{equation}
where $\varepsilon$ is an arbitrarily small positive number. Therefore, the boundary (\ref{bound2}) is reformulated for $m\geq m^{\star}$ as
\begin{equation}
\label{bound3}
\begin{aligned}
\mathbb{E}\left\{\left\|\bX_{i}-\bX^{\star}\right\|_{\mathrm{F}}^{2}\right\}& \leq q^{i}\left\|\bX_{0}-\bX^{\star}\right\|_{\mathrm{F}}^{2}+\epsilon_{0}e^{-\gamma_{1} m}.
\end{aligned}
\end{equation}

\subsection{Lower Bound on the Number of Measurements}
The algorithm termination criterion is considered to be
\begin{equation}
\label{eq:6565}
\mathbb{E}\left\{\left\|\bX_{i}-\bX^{\star}\right\|_{\mathrm{F}}^{2}\right\}\leq \epsilon_{1}.
\end{equation}
Based on this criterion and (\ref{bound3}), the following theorem is presented in order to find a lower bound for the number of required measurements in OPeRA.
\begin{theorem}
\label{theorem_1}
To recover the desired PSD matrix $\bX^{\star}$ in accordance to (\ref{eq:6565}) in one-bit phase retrieval by OPeRA
with a probability of at least $1-\inf\left\{M_{T}e^{-at}\right\}$, the number of measurements $m$ must obey
\begin{equation}
\label{bound5}
\begin{aligned}
m&\geq \frac{1}{\gamma_{1}} \ln\left(\frac{\epsilon_{0}}{\epsilon_{1}-q^{i}\omega_{0}}\right),
\end{aligned}
\end{equation}
where $\epsilon_{0}$ and $\gamma_{1}$ are determined via (\ref{eq:6585}), $\omega_{0}=\left\|\bX_{0}-\bX^{\star}\right\|_{\mathrm{F}}^{2}$ is the initial squared-error, and $i$ is the number of iterations.
\end{theorem}
\begin{proof}
To satisfy (\ref{eq:6565}), the inequality in (\ref{bound3}) is adjusted to capture the upper bound $\epsilon_{1}$. As a result, the following inequality is obtained:
\begin{equation}
\label{bound31}
\begin{aligned}
q^{i}\left\|\bX_{0}-\bX^{\star}\right\|_{\mathrm{F}}^{2}+\epsilon_{0}e^{-\gamma_{1} m}\leq \epsilon_{1}.
\end{aligned}
\end{equation}
Note that $\omega_{0}$ is a constant scalar that only depends on the initial and optimal signals. According to (\ref{bound3}), we can write
\begin{equation}
\label{bound30}
e^{-\gamma_{1} m}\leq \frac{\epsilon_{1}-q^{i}\omega_{0}}{\epsilon_{0}},
\end{equation}
or equivalently,
\begin{equation}
\label{iiii}
m\geq \frac{1}{\gamma_{1}}\ln\left(\frac{\epsilon_{0}}{\epsilon_{1}-q^{i}\omega_{0}}\right).
\end{equation}
which proves the theorem.
\end{proof}
The sufficient number of measurements $m$ is sought to create a finite-volume in the PSD cone $\bX \succeq 0$ for the input signal $\mathbf{x}$ with the size $n$. It is easy to verify that the dimension of the PSD cone is equal to $\frac{n^{2}+n}{2}$ in both real-valued and complex valued cases. 
The infimum number of the hyperplanes creating a finite-volume space in a $n^{\prime}$-dimensional region is $n^{\prime}+1$. Consequently, we can write:
\begin{equation}
\inf \mathcal{K}= \frac{n^{2}+n}{2}+1,\quad \mathcal{K}\subseteq\mathcal{J}.
\end{equation}
Therefore, to form the finite-volume space, the sample size $m$ must be lower bounded by $\inf\mathcal{K}$:
\begin{equation}
\label{caom}
m\geq \left(\frac{n^{2}+n}{2}+1\right).
\end{equation}
The above bound helps us to establish a clear connection between the sample size and the problem dimension. However, since the boundary in (\ref{bound5}) is concerned with containing the error after a finite-volume is formed, it is asymptotically tighter than (\ref{caom}). Therefore, by satisfying (\ref{bound5}), the lower bound in (\ref{caom}) is typically met as well.
\section{Complexity Investigation}
\label{sec_com}
We examine the computational cost associated with the use of more samples by comparing our approach (OPeRA) with the PhaseLift method and its one-bit extension (one-bit PhaseLift) as defined in (\ref{eq:6}). This comparison is based on the required computational time for different sample sizes.
\subsection{Comparing PhaseLift and OPeRA}
\label{sec_com_sdp}
To compare the computational time for the SDP-based PhaseLift approach and our proposed method, the iterative algorithms are terminated according to (\ref{eq:6565}) with $\epsilon_{1}=10^{-2}$. The unknown signal $\mathbf{x}$ lies within $\mathbb{R}^{10}$. The number of samples $m$ is set to be $100$, $1000$, $3000$, $4000$, $5000$, $10000$, $30000$, $50000$, $80000$, and $100000$. Each CPU time is obtained by averaging over $5$ experiments. PhaseLift is applied to the high-resolution samples, whereas OPeRA is applied to their one-bit sampled data counterpart, which means only partial information is made available to OPeRA.

As can be seen in Fig.~\ref{figure_6}, due to the growing number of samples, the cost of the PhaseLift algorithm has an increasing trend. Nevertheless, the CPU time for OPeRA experiences a significant decline rate from $m=100$ up until $m=50000$, while it starts increasing afterwards. The reason behind this behavior is hidden in the application of RKA. As discussed in Section~\ref{error_bound}, to create a finite-volume space around the optimal point $\bX^{\star}$, and to capture error upper bound $\epsilon_{1}$, the number of samples has to move to the large-scale regimen. One may simply verify that, according to Theorem~\ref{theorem_1}, by increasing the number of measurements $m$, the RKA may achieve the error upper bound with fewer iterations $i$. Additionally, the computational cost of the RKA used to solve (\ref{eq:25}) behaves as $\mathcal{O}\left(i n^{2}\right)$. As a result, the CPU time for OPeRA may initially decrease with a growing number of samples.

Surpassing $m=3000$, the CPU time for OPeRA becomes smaller than that of the PhaseLift method; i.e. by employing sufficient number of samples, less complexity is achievable which is facilitated by dropping the PSD constraint. However, after approaching $m=50000$, we have increased the number of inequalities in such a way that the contribution to signal recovery is negligible, while at the same time, the extra measurements will still need processing. These extra inequalities may require more iterations to take them into account which is undesirable and increases the CPU time; see Fig.~\ref{figure_6}. Interestingly, OPeRA can satisfy the recovery criterion (\ref{eq:6565}) with less CPU time and less input information compared to PhaseLift which is useful for the high-resolution scenario.
\begin{figure}[t]
	\center{\includegraphics[width=0.6\textwidth]{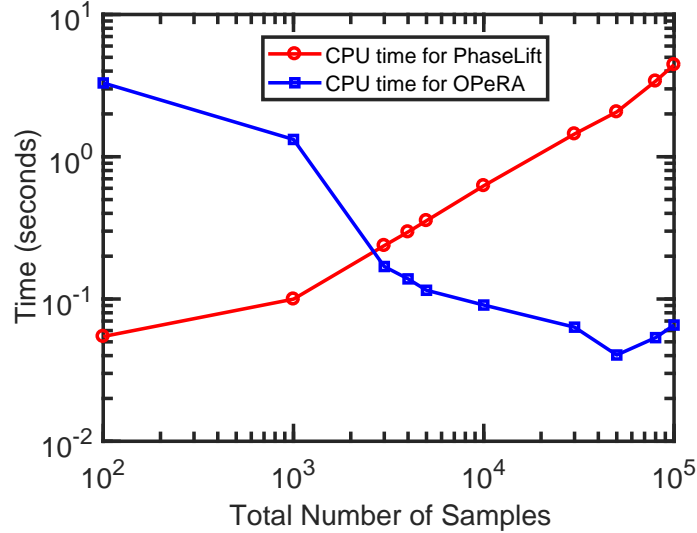}}
	\caption{Comparing the PhaseLift method and OPeRA per CPU time when both methods use the stopping criterion $\mathbb{E}\left\{\left\|\bX_{i}-\bX^{\star}\right\|_{\mathrm{F}}^{2}\right\}\leq 10^{-2}$. Note that PhaseLift is applied to the high-resolution samples whereas only the one-bit version of the samples are made available to OPeRA.}
	\label{figure_6}
\end{figure}
\subsection{Comparing One-bit PhaseLift and OPeRA}
\label{sec_com_LP}
Next, we compare the CPU time of the one-bit PhaseLift defined in (\ref{eq:6}) and that of OPeRA. As mentioned before, the one-bit PhaseLift problem is a SDP similar to the problem (\ref{eq:1nnnn}).

To compare the relaxation-based formulation of the one-bit PhaseLift with our approach, we aim at the recovery of a signal $\mathbf{x}\in \mathbb{R}^{10}$ with sample sizes $m \in \left\{1000, 3000, 5000, 8000, 10000\right\}$. The termination criterion of both algorithms is exactly similar to the one in Section~\ref{sec_com_sdp}.

As can be seen in Fig.~\ref{figure_10}, by growing the number of one-bit samples, the CPU time of the one-bit PhaseLift is increasing. On the other hand, the CPU time of OPeRA is decreasing. One may conclude that in the large sample size regimen, the proposed one-bit phase retrieval approach has a lower computational burden than one-bit PhaseLift. This due to enjoying the advantages of using more samples to make both the rank-one constraint and the PSD constraint redundant. We hypothesis that the computation time rise due to unhelpful extra samples (in achieving the error bound) occur well beyond $m=10000$, which is presumably why the increase is not observed in our experiment.
\section{Adaptive Time-Varying Thresholding}
\label{sec_adaptive}
Hereafter, we propose an adaptive threshold design strategy for the task of one-bit phase retrieval. By the spirit of using the iterative RKA, a suitable time-varying threshold can be chosen in order to find the optimal solution with enhanced accuracy. As discussed earlier, with sample abundance, we have an overdetermined linear system of inequalities creating a finite-volume space. To capture the desired signal matrix $\bX^{\star}$ more efficiently, the right-hand side of the inequalities in (\ref{eq:25}), i.e. $r_{j}\tau^{2}_{j}$, must be determined in a way that each associated hyperplane passes through the desired feasible region within the PSD cone. Therefore, an algorithm is proposed to ensure that this occurs. To give an illustration, we suppose the solution is the yellow point in Fig.~\ref{figure_tresh}. Geometrically, with an adaptive time-varying threshold algorithm, our goal is to generate an informative sampling threshold creating the inequality constraint corresponding to the hyperplane shown by the blue line.

Unlike the other two inequality constraints (hyperplanes illustrated by green and  purple lines in Fig.~\ref{figure_tresh}), the blue hyperplane will further shrink the feasible region for signal recovery. From this viewpoint, the other hyperplanes (green and purple lines) constitute extra inequality constraints that are not informative. As discussed in Section~\ref{sec_com_sdp}, such extra samples (the extra inequality constraints) only increase the computational burden of the phase retrieval task.
\begin{figure}[t]
	\center{\includegraphics[width=0.6\textwidth]{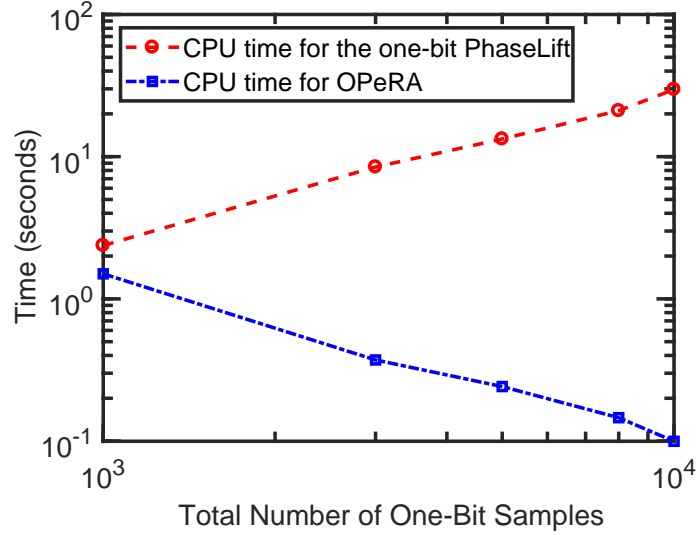}}
	\caption{Comparing the one-bit PhaseLift method and OPeRA based on their CPU time when both algorithms are terminated using the criterion $\mathbb{E}\left\{\left\|\bX_{i}-\bX^{\star}\right\|_{\mathrm{F}}^{2}\right\}\leq 10^{-2}$.}
	\label{figure_10}
\end{figure}
\subsection{Adaptive Threshold Design for OPeRA}
\label{sec:adaptive_algorithm_design}
In light of previous discussion, to achieve a better recovery accuracy for a specific sample size $m$, one may use the idea of shrinking the space imposed by the set of inequalities in (\ref{eq:25}) around the optimal solution $\bX^{\star}$. To make this happen, we propose an iterative algorithm generating an adaptive threshold to accurately obtain the desired solution. To diminish the area of the finite-volume space around the optimal point, we update the time-varying threshold as
\begin{equation}
\label{eq:200}
\operatorname{Tr}\left(\bV_{j}\bX^{(k)}\right)-r_{j}^{(k)}\epsilon^{(k)}_{j} = \left(\tau^{(k+1)}_{j}\right)^{2}, \quad j\in \mathcal{J},
\end{equation}
where $\left\{\epsilon_{j}^{(k)}\right\}$ at the $j$-th element of a positive vector $\bepsilon^{(k)}$ in the $k$-th iteration. This updating process is based on the fact that when $r_{j}=+1$, we have $\operatorname{Tr}\left(\bV_{j}\bX\right)\geq (\tau_{j})^{2}$, and $\operatorname{Tr}\left(\bV_{j}\bX\right)\leq (\tau_{j})^{2}$ otherwise. The reason behind updating the one-bit measurements $\left\{r_{j}^{(k)}\right\}$, is to ensure that the optimal solution $\bX^{\star}$ satisfies (\ref{eq:25}) in iteration $k$, i.e. the inequalities $r_{j}^{(k)} \operatorname{vec}\left(\bV^{\top}_{j}\right)^{\top}\operatorname{vec}\left(\bX^{\star}\right) \geq r_{j}^{(k)}\left(\tau_{j}^{(k)}\right)^{2}$ for $j\in \mathcal{J}$. Our proposed iterative method to update the time-varying threshold is summarized in Algorithm~\ref{algorithm_1}.
\begin{figure}[t]
	\center{\includegraphics[width=0.6\textwidth]{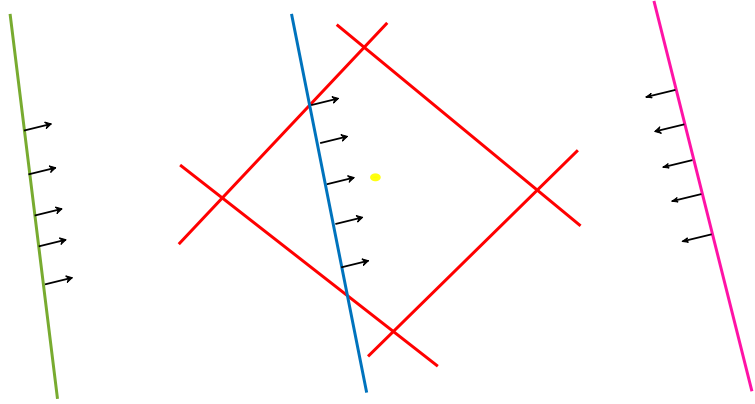}}
	\caption{The signal of interest  (shown by yellow circle) is located in the finite-volume space which presents itself as the red semi-rectangular. The new shrunken feasible region created by the blue hyperplane captures the desired signal. However, the other two sampling hyperplanes are non-informative.}
	\label{figure_tresh}
\end{figure}
\subsection{Numerical Study of Adaptive Thresholding}
\label{num_adaptive}
To present the efficacy of the adaptive time-varying threshold algorithm in comparison with a non-negative random threshold, the unknown signal $\mathbf{x}$ and the sensing matrix $\bA$ are generated using the same settings as in Section~\ref{NUM}. As can be seen in Fig.~\ref{figure_7}, the performance of our proposed algorithm is evaluated by the NMSE defined in (\ref{eq:4000}) which is considerably enhanced in comparison with adopting a non-negative random threshold according to $\operatorname{Lognormal}\left(0,1\right)$. The results are obtained for the sample sizes $m\in \left\{1000, 5000, 10000, 50000, 100000\right\}$
and $\delta=10^{-3}$. Each presented data point is averaged over $5$ experiments.

In the following, we compare the number of measurements $m$ required to recover the rank-one and PSD matrix $\bX^{\star}$ using the OPeRA with (i) our proposed adaptive threshold algorithm, as well as (ii) a random non-negative threshold. The result can be summarized as follows.
\begin{theorem}
\label{lemma_1}
OPeRA with the adaptive sampling threshold proposed in Algorithm~\ref{algorithm_1} can recover the rank-one and PSD matrix $\bX^{\star}$ with a high probability of at least $1-\inf\left\{M_{T}e^{-at}\right\}$ by using a smaller number of measurements $m$ in comparison to OPeRA with a random threshold.
\end{theorem}
\begin{proof}
\label{proof_adaptive}
Let $m_{f}$ denote the number of inequalities in (\ref{eq:25}) creating a finite-volume space around the desired matrix $\bX^{\star}$, which is not necessarily inside the PSD cone. As discussed in Section~\ref{sec:adaptive_algorithm_design}, our proposed adaptive thresholding algorithm shrinks the finite-volume space around the optimal solution $\bX^{\star}$ in a stronger way than the random thresholds. As a result, $\left\{H_{j}\right\}$ defined in (\ref{distance}) as well as their moments $\mu^{(k)}_{H_{j}}$, will further diminish which leads to a smaller value for $M_{T}$. Therefore, according to Theorem~\ref{theorem_0} and Theorem~\ref{theorem_1}, a similar recovery performance can be expected with a smaller sample size when the adaptive sampling threshold proposed in Algorithm~\ref{algorithm_1} is utilized.
\end{proof}
To numerically scrutinize our claim in Theorem~\ref{lemma_1}, Table~\ref{table_1} illustrates that the sufficient number of samples $m$ required for OPeRA to recover a PSD matrix with adaptive sampling thresholds proposed in Algorithm~\ref{algorithm_1} is much less than that of OPeRA with a random threshold. The result is obtained for the input signal $\mathbf{x}$, a random time-varying threshold $\btau$, and the sensing matrix $\bA$ originating from the same settings as presented in Section~\ref{NUM}. The number of samples examined in our experiments are same as Section~\ref{sec_com_sdp}, plus $m=500$.
\begin{algorithm}[t]
\SetAlgoLined
\emph{Input:} One-bit data: $\br$, sensing matrix: $\bA$, initial time-varying thresholds: $\btau^{(0)}\sim \operatorname{Lognormal}\left(0,1\right)$ (with the same length as $\br$), small positive number: $\delta$.\\
\emph{Output:} Adaptive threshold: $\btau^{\star}$.\\
\emph{Note:} $\bX^{(k)}$, $\btau^{(k)}$, $\br^{(k)}$ and $\bepsilon^{(k)}$ denote their associated values at iteration $k$.\\
\SetAlgoLined

- Set $\br^{(0)}=\br$.\\
- Calculate the matrix $\bV$ with the rows $\{\operatorname{vec}\left(\ba_{j}\ba^{\mathrm{H}}_{j}\right)^{\top}\}$.\\
- Initiate the following loop by setting $k=0$.\\
\While {$\left\|\btau^{(k+1)}-\btau^{(k)}\right\|_{2}\leq \delta$}{
- Find a point inside the following polyhedron with the RKA for $\btau=\btau^{(k)}$:
\[
\mathcal{P}=\left\{\bX^{(k)} \mid \left(\bR^{(k)}\odot \bV\right)\operatorname{vec}\left(\bX^{(k)}\right)\geq\bb^{(k)}\right\},
\]
where $\bR^{(k)}$ denotes the replica form of $\br^{(k)}$ and $\bb^{(k)}=\br^{(k)} \odot \left(\btau^{(k)}\right)^{2}$.\\
- Update $\btau^{(k+1)}$ as:
\[
\br^{(k)} \odot \left(\btau^{(k+1)}\right)^{2}=\left(\bR^{(k)}\odot \bV\right)\operatorname{vec}\left(\bX^{(k)}\right)-\frac{\bepsilon^{(k)}}{2},
\]
where $\bepsilon^{(k)}$ is computed as:
\[
\bepsilon^{(k)}=\left(\bR^{(k)}\odot \bV\right)\operatorname{vec}\left(\bX^{(k)}\right)-\br^{(k)} \odot \left(\btau^{(k)}\right)^{2}.
\]
- Update $\br^{(k+1)}$ based on (\ref{eq:11}).\\
- Increase $k$ by one.
}
\caption{Adaptive Algorithm for Sampling Threshold Selection}
\label{algorithm_1}
\end{algorithm}

To further investigate the effectiveness of Algorithm~\ref{algorithm_1}, we show that the average of all eigenvalues $\{\ell_{i}\}$ of $\bar{\bX}$ except the maximum eigenvalue (spectral radius) approaches zero by increasing the number of measurements---see Fig. \ref{eigenvalue_adaptive}. Interestingly, Fig.~\ref{eigenvalue_adaptive} reaffirms that the number of measurements used to recover the rank-one and PSD matrix $\bX^{\star}$ by OPeRA with adaptive thresholding, is much less than the same result reported in Fig.~\ref{figure_3} where OPeRA with a random threshold is adopted. The results are obtained for the number of samples $m\in\left\{300, 500, 1000, 2000\right\}$, they are averaged over $5$ experiments, and the non-dominant eigenvalues are arranged in a decreasing order.
\begin{figure}[t]
	\center{\includegraphics[width=0.6\textwidth]{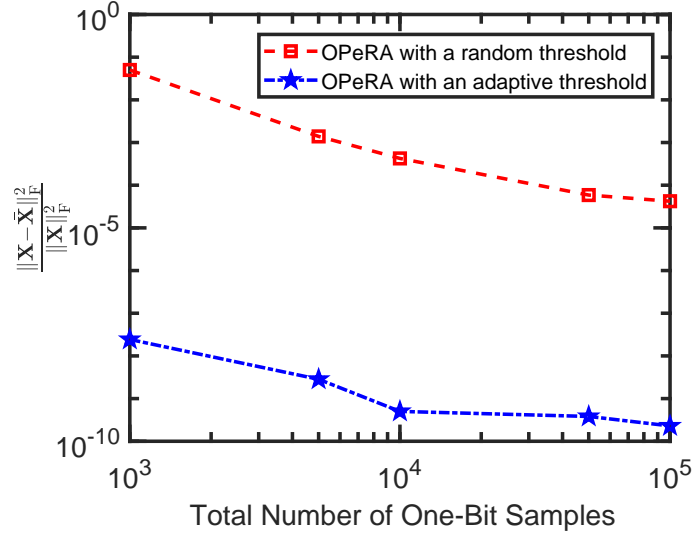}}
	\caption{Comparing the average NMSE for the Frobenius norm between the desired matrix $\bX^{\star}$ and its recovered matrix using OPeRA when (i) a random threshold and (ii) the adaptive sampling threshold proposed in Algorithm~\ref{algorithm_1}, are adopted.}
	\label{figure_7}
\end{figure}

\section{One-Bit Phase Retrieval with Noisy Measurements}
\label{sec_denoising}
In this section, we extend our study to signal recovery from noisy one-bit data in the phase retrieval problem. In most practical applications, we must rely on noisy measurements \cite{candes2015phase,davenport20141,bhaskar20151}. In particular, we will examine whether the computational advantages provided by sample abundance in the noiseless scenario will also be observed under the presence of noise.

\subsection{Problem Formulation}
\label{Sec_prob_noisy}
Define the positive-valued vector $\bmu=\left[\mu_{1},\cdots,\mu_{m}\right]$ by
\begin{equation}
\label{eq:100}
\mu_{j} = \operatorname{Tr}\left(\bV_{j} \bX\right),\quad j\in \mathcal{J}.
\end{equation}
Let $\blambda$ and $\mathbf{z}$ denote the time-varying threshold vector and the noise vector, respectively. The noisy one-bit samples are generated as
\begin{equation}
\label{eq:101}
\begin{aligned}
r_{j} &= \begin{cases} +1 & \mu_{j}+z_{j}>\lambda_{j}, \\ -1 & \mu_{j}+z_{j}<\lambda_{j}.
\end{cases}
\end{aligned}
\end{equation}
The occurrence probability vector $\bp$ for the noisy one-bit measurement $\br$ is given as \cite{bhaskar20151},
\begin{equation}
\label{eq:102}
\begin{aligned}
p_{j} &= \begin{cases} \Phi(\mu_{j}-\lambda_{j}) & \text{for}\quad \{r_{j}=+1\}, \\ 1-\Phi(\mu_{j}-\lambda_{j}) &  \text{for}\quad \{r_{j}=-1\},
\end{cases}
\end{aligned}
\end{equation}
where $\Phi(.)$ is the CDF of $-\mathbf{z}$. Since $\{\mu_{j}\}$ are linear function of $\bX$, the CDF of noise $\left\{\Phi(\mu_{j}-\lambda_{j})\right\}$ can be written as $\Phi\left(\bX\right)$. The log-likelihood function of the sign data $\br$ is given by
\begin{equation}
\label{eq:103}
\begin{aligned}
\mathcal{L}_{\br}(\bmu,\bX) &= \sum^{m}_{j=1}\left\{\mathbb{I}_{(r_{j}=+1)}\log\left(\Phi(\mu_{j}-\lambda_{j})\right) \right.\\& \left.+\mathbb{I}_{(r_{j}=-1)}\log\left(1-\Phi(\mu_{j}-\lambda_{j})\right)\right\} .
\end{aligned}
\end{equation}
Interestingly, by solving the maximum log-likelihood estimation (MLE) problem associated with (\ref{eq:103}), our desired matrix $\bX^{\star}$ can be immediately approximated. The proposed algorithm is called \emph{Noisy OPeRA}.
\begin{table} [t]
\caption{The number of samples required to recover a PSD matrix.}
\centering
\begin{tabular}{ | c | c | }
\hline
\text {Proposed Algorithm} & \text {$m$} \\[0.5 ex]
\hline \hline
\text{OPeRA with a random threshold} & $30000$ \\[1 ex]
\hline
\text{OPeRA with the adaptive threshold} & $500$ \\[1 ex]
\hline
\end{tabular}
\label{table_1}
\end{table}
\begin{figure}[t]
	\center{\includegraphics[width=0.6\textwidth]{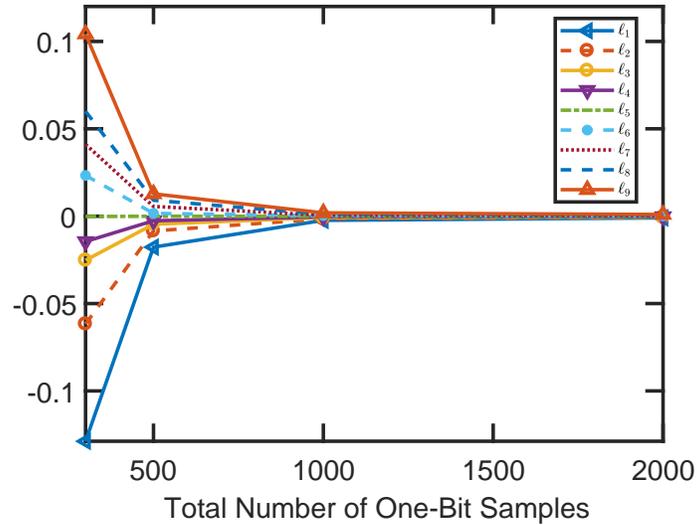}}
	\caption{All eigenvalues $\{\ell_{i}\}$ of $\bar{\bX}$ except the dominant eigenvalue (spectral radius), averaged over $5$ experiments. As can be seen, deploying OPeRA with the adaptive sampling thresholds, leads to obtaining a nearly rank-one and PSD matrix with significantly enhanced accuracy as the number of samples grows large.}
	\label{eigenvalue_adaptive}
\end{figure}
\subsection{Noisy One-Bit Phase Retrieval via Convex Programming}
\label{Sec_convex_noisy}
A preliminary formulation of our optimization problem based on the MLE may be cast as:
\begin{equation}
\label{eq:104}
\begin{aligned}
\min_{\bmu,\bX}  \quad &-\mathcal{L}_{\br}(\bmu,\bX)\\
\text{s.t.}\quad &\mu_{j} = \operatorname{Tr}\left(\bV_{j}\bX\right), \quad j\in \mathcal{J},\\ 
&\operatorname{rank}\left(\bX\right)=1,\\
&\bX \succeq 0.
\end{aligned}
\end{equation}
This problem is the one-bit version of its counterpart formulated in \cite{candes2015phase}. However, as discussed in previous sections, because of employing one-bit sampling, the large number of samples can be adopted which leads to the availability of a large number of sign data $\left\{r_{j}\right\}$ and the corresponding inequality constraints; since when $r_{j}=+1$, we have $\operatorname{Tr}\left(\bV_{j}\bX\right)\geq (\tau_{j})^{2}$, and $\operatorname{Tr}\left(\bV_{j}\bX\right)\leq (\tau_{j})^{2}$ otherwise. These inequalities are collected to form the polyhedron (\ref{eq:25}). However, these constraints may be equivalently absorbed in the objective function to facilitate the one-bit phase retrieval formulation in the noisy case. Therefore, the problem (\ref{eq:104}) can be reformulated as
\begin{equation}
\label{eq:104}
\begin{aligned}
\min_{\bmu,\bX}  \quad &-\mathcal{L}_{\br}(\bmu,\bX)\\
\text{s.t.}\quad &\mu_{j} = \operatorname{Tr}\left(\bV_{j}\bX\right), \quad j\in \mathcal{J}.
\end{aligned}
\end{equation}
In many cases, $\mathcal{L}_{\br}(\bmu,\bX)$ is a concave function and thus the above programs becomes convex. One can readily verify this in the case of a Gaussian noise \cite{davenport20141}. In the rest of our paper, $-\mathbf{z}$ is assumed to be an i.i.d. zero-mean Gaussian process $\mathbf{z}\sim \mathcal{N}\left(0,\sigma^{2}_{\mathbf{z}} \bI_{m}\right)$, for which $\Phi(.)$ is given in (\ref{eq:1bbb}).
\begin{figure}[t]
	\center{\includegraphics[width=0.6\textwidth]{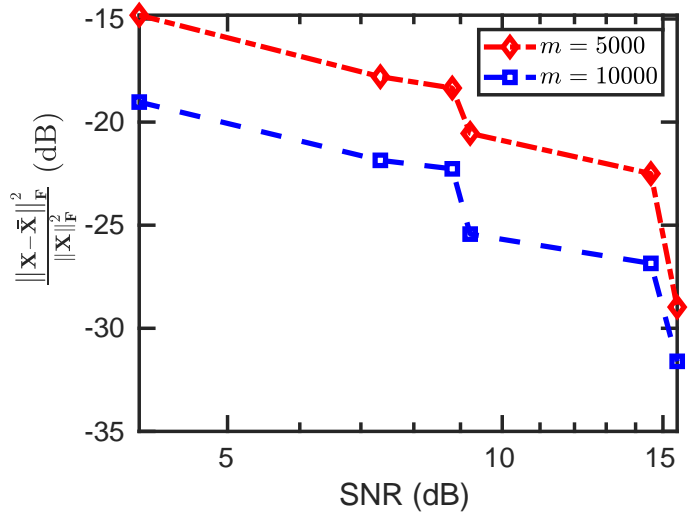}}
	\caption{Average NMSE (dB) in the results obtained by the MLE problem (\ref{eq:104}) over different SNRs and for two different sample sizes $m\in\left\{5000, 10000\right\}$. It is observed that by increasing SNR and the number of one-bit measurements, the recovery performance of Noisy OPeRA is enhanced.}
	\label{figure_8}
\end{figure}

\subsection{Numerical Investigation of Noisy OPeRA}
\label{Sec_nr_noisy}
To examine the performance of Noisy OPeRA in practice, and to validate the theoretical results described in this section, we consider signal recovery with different values of  $\sigma_{\mathbf{z}}\in \left\{0.1, 0.2, 0.4, 0.5, 0.7, 1\right\}$, where the unknown signal $\mathbf{x}$ was generated in a similar manner as in Section~\ref{NUM}. The stochastic threshold $\blambda$ was generated according to  $\blambda \sim \mathcal{N}\left(0,I_{m}\right)$. The signal to noise ratio (SNR) is evaluated as:
\begin{equation}
\label{eq:1001}
\mathrm{SNR} = \frac{\frac{1}{m}\sum^{m}_{j=1}\mu^{2}_{j}}{\sigma^{2}_{\mathbf{z}}}.
\end{equation}
In Fig.~\ref{figure_8}, the recovery performance is illustrated by using the NMSE defined in (\ref{eq:4000}), with the results averaged over $10$ experiments. We report both SNR and NMSE in dB ($10\log(.)$). As expected, by increasing the SNR, the performance of our method is improved. Furthermore, the performance of the estimation problem formulation in (\ref{eq:104}) is enhanced by increasing the number of one-bit samples $m\in \{5000,10000\}$. In this approach, since the desired matrix $\bX^{\star}$ is recovered statistically from MLE, we compare $\Phi\left(\bX\right)$ and $\Phi(\bar{\bX})$ by resorting to a widely used statistical distance, known as the Hellinger distance, which was defined in (\ref{eq:1bbbb}). The vector entry-wise formula of the Hellinger distance is given as
\begin{equation}
\label{eq:10000}
\bd^{2}_{H}\left(\Phi\left(\bX\right),\Phi(\bar{\bX})\right)=\frac{1}{m}\sum^{m}_{j=1} d^{2}_{H}\left(\Phi\left(\mu_{j}-\lambda_{j}\right),\Phi(\bar{\mu}_{j}-\lambda_{j})\right),
\end{equation}
where $\{\bar{\mu}_{j}\}$ is the estimated version of $\{\mu_{j}\}$ obtained from (\ref{eq:104}). As was previously observed, by increasing the value of SNR, Noisy OPeRA performs better in terms of the NMSE . A similar behavior occurs with the Hellinger distance shown in Fig.~\ref{figure_9}. The Hellinger distance is obtained is very small for all SNR values in this experiment, However, it is decreasing for an increasing SNR, which appears to confirm the recovery of the desired matrix in statistical (noisy) environments by taking advantage of a large number of samples---thus without considering the rank-one and the PSD constraints.
\begin{figure}[t]
	\center{\includegraphics[width=0.6\textwidth]{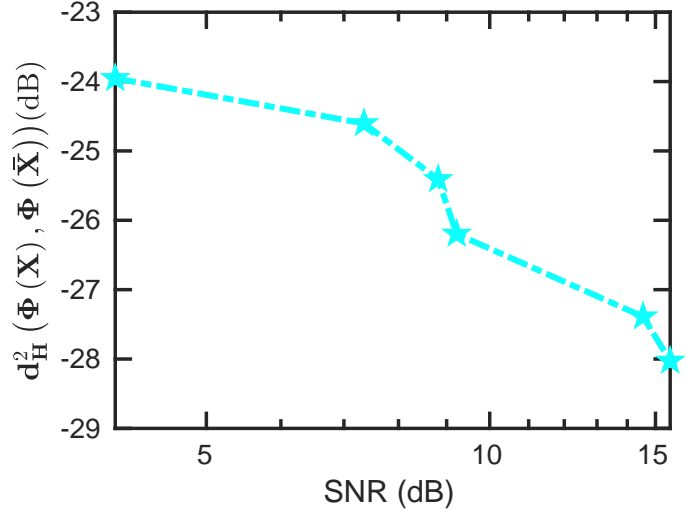}}
	\caption{Comparing the CDF of the desired matrix $\Phi\left(\bX\right)$ and the CDF of the recovered matrix $\Phi\left(\bar{\bX}\right)$ using the Hellinger distance (\ref{eq:10000}). Although the Hellinger distance of our estimation is very small overall, it shows a decreasing behavior as the SNR grows large.}
	\label{figure_9}
\end{figure}

To show the sustained benefit of sample abundance in the noisy case, we compare Noisy OPeRA with \emph{Noisy PhaseLift} formulation firstly introduced in \cite{candes2015phase} as
\begin{equation}
\label{eq:104000}
\begin{aligned}
\min_{\bmu,\bX}  \quad &-\Upsilon_{\bmu}(\bX)+\alpha \operatorname{Tr}\left(\bX\right) \\
\text{s.t.}\quad &\mu_{j} = \operatorname{Tr}\left(\bV_{j}\bX\right), \quad j\in \mathcal{J},\\
&\bX \succeq 0,
\end{aligned}
\end{equation}
where $\Upsilon_{\bmu}(\bX)=\log\left(f\left(\bs|\bmu\right)\right)$, with the noisy measurement vector $\{s_{j}\}$ is sampled from a probability distribution $f(.|\bmu)$, and $\alpha$ is a positive scalar. For our numerical examinations, we assume the measurement noise is distributed as $\mathbf{z}\sim \mathcal{N}\left(0,0.25 \bI_{m}\right)$ and the termination criterion is $\left\|\bX_{i}-\bX^{\star}\right\|_{\mathrm{F}}^{2}\leq 5\times 10^{-3}\left\|\bX^{\star}\right\|_{\mathrm{F}}^{2}$. Table~\ref{table_2} shows that by using a large number of samples (and making rank-one and PSD constraints redundant) in the noisy one-bit sampling scenario, Noisy OPeRA can recover the signal with a better CPU time for sample sizes $m\in\left\{5000, 10000, 20000\right\}$ compared to the noisy PhaseLift method. This is similar to our discussion in the noiseless scenario; see Section~\rom{5}. Interestingly, by growing the number of samples, the NMSE is enhanced more significantly by Noisy OPeRA than that of the noisy PhaseLift method. The results are averaged over $5$ experiments. The settings of the input signal, time-varying thresholds and the sensing are also chosen in the same way as in Section~\ref{NUM}.
\begin{table} [t]
\centering
\caption{Comparing Noisy PhaseLift and Noisy OPeRA in terms of CPU time and NMSE.}
\centering
\begin{tabular}{ | c | c | c | c |}

\hline
\text{Noisy PhaseLift\cite{candes2015phase}} & \text {$m=5000$}  & \text {$m=10000$}& \text {$m=20000$}\\ [0.5 ex]
\hline \hline
CPU time ($s$) & 1.4698 &  2.0305& 3.7529\\
\hline
NMSE & 0.0045 & 0.0041& 0.0035\\
\hline\hline
\text{Noisy OPeRA}& \text {$m=5000$}  & \text {$m=10000$}& \text {$m=20000$}\\
\hline \hline
CPU time ($s$) & 0.9497 &  1.4436& 2.3137\\
\hline
NMSE & 0.0040 &  0.0015& 3.8875e-04 \\
\hline
\end{tabular}
\label{table_2}
\end{table}
\section{Conclusion}
We showed that the abundance of samples that naturally occurs in one-bit sampling scenarios has significant implications in lowering the computational cost of phase retrieval by making costly constraint redundant. The problem then boils down to a set of linear inequalities that may be solved by RKA within the proposed OPeRA signal recovery framework. The numerical results showcased the effectiveness of the proposed approaches for phase retrieval.

\bibliographystyle{IEEEbib}
\bibliography{strings,refs}

\end{document}

%% file: symbols.tex
\usepackage{amsmath}
\usepackage{algpseudocode}
\usepackage[lined,boxed,commentsnumbered]{algorithm2e}
\usepackage{amsthm}
\usepackage{dsfont}
\usepackage{thmtools}
\usepackage{amssymb}
\usepackage[dvipsnames]{xcolor}

\def\ba{\boldsymbol{a}}
\def\bb{\boldsymbol{b}}
\def\bc{\boldsymbol{c}}
\def\bd{\boldsymbol{d}}
\def\be{\boldsymbol{e}}

\def\bp{\boldsymbol{p}}

\def\br{\boldsymbol{r}}
\def\bs{\boldsymbol{s}}

\def\bz{\boldsymbol{z}}

\def\bA{\boldsymbol{A}}
\def\bB{\boldsymbol{B}}
\def\bC{\boldsymbol{C}}
\def\bD{\boldsymbol{D}}

\def\bH{\boldsymbol{H}}
\def\bI{\boldsymbol{I}}

\def\bR{\boldsymbol{R}}

\def\bV{\boldsymbol{V}}

\def\bX{\boldsymbol{X}}

\def\bepsilon{\boldsymbol{\epsilon}}

\def\blambda{\boldsymbol{\lambda}}

\def\bmu{\boldsymbol{\mu}}

\def\btau{\boldsymbol{\tau}}

\newtheorem{theorem}{Theorem}

\theoremstyle{definition}

\algnewcommand\algorithmicinput{\textbf{Input:}}
\algnewcommand\Input{\item[\algorithmicinput]}
\algnewcommand\algorithmicoutput{\textbf{Output:}}
\algnewcommand\Output{\item[\algorithmicoutput]}
\algnewcommand\algorithmicinit{\textbf{Initialize:}}
\algnewcommand\Init{\item[\algorithmicinit]}
